%% file: main.tex
\documentclass[letterpaper, 11pt]{article}
\usepackage[english]{babel}
\usepackage{inputenc}
\usepackage{amsfonts}
\usepackage{mathtools,amssymb,amsmath,amsthm,pifont}
\usepackage{mathrsfs}
\usepackage{graphicx}
\usepackage{float}
\usepackage{autobreak}
\usepackage{bm,bbm}
\usepackage{mathtools}
\usepackage{thm-restate}
\usepackage{mathrsfs}
\usepackage[left=1.125in,right=1.125in,top=1in,bottom=1in]{geometry}
\usepackage{todonotes}
\usepackage[colorlinks]{hyperref}
\usepackage{cleveref}
\usepackage{graphicx} 
\newcommand{\B}{\mathcal{B}}

\newcommand{\F}{\mathbb{F}}
\newcommand{\N}{\mathbb{N}}

\theoremstyle{plain}
\newtheorem{theorem}{Theorem}[section]

\newtheorem{lemma}[theorem]{Lemma}
\newtheorem{proposition}[theorem]{Proposition}

\theoremstyle{definition}
\newtheorem{definition}[theorem]{Definition}

\theoremstyle{remark}
\newtheorem{remark}[theorem]{Remark}

\newcommand{\cC}{\ensuremath{{\mathcal C}}} 
  
\newcommand{\cM}{\ensuremath{\mathcal{M}}}

\newcommand{\cP}{\ensuremath{\mathcal{P}}}

\newcommand{\cD}{\ensuremath{\mathcal{D}}}

\newcommand{\cB}{\mathcal{B}}

\newcommand{\supp}{\mathrm{supp}}

\newcommand{\bv}{\bm{v}}

\newcommand{\bx}{\bm{x}}
\newcommand{\by}{\bm{y}}
\newcommand{\bz}{\bm{z}}

\newcommand{\be}{\bm{e}}
\newcommand{\bc}{\bm{c}}
\newcommand{\bS}{\bm{S}}

\newcommand{\Eop}{\mathop{\mathbb E}}

\newcommand{\TRLC}{R_{\mathrm{RLC}}}

\newcommand{\RLC}{\mathrm{RLC}}
\newcommand{\PCLP}{\mathrm{PCLP}}

\newcommand{\Emp}{\mathsf{Emp}}
\newcommand{\PCRCP}{\mathrm{PCRCP}}

\newcommand{\xmark}{\ding{55}}%

\newcommand{\poly}{\mathrm{poly}}

\newcommand{\rank}{\ensuremath{\operatorname{rank}}}

\newcommand{\wt}{\mathrm{wt}}

\newcommand{\eps}{\varepsilon}
\renewcommand{\epsilon}{\varepsilon}

\newcommand{\suchthat}{:}
\DeclareMathOperator{\Tr}{Tr}
\newcommand{\inparen}[1]{\left(#1\right)}

\DeclareMathOperator{\eperiod}{\,\,\text{.}}

 \newcommand{\jm}[1]{}
 \newcommand{\jmi}[1]{}
 \newcommand{\nr}[1]{}
 \newcommand{\nri}[1]{}
 \newcommand{\chen}[1]{}

\begin{document}

\title{Randomness-Efficient Constructions of Capacity-Achieving List-Decodable Codes}
\author{Jonathan Mosheiff\thanks{Department of Computer Science, Ben-Gurion University. Research supported by an Alon Fellowship. Part of this work was conducted while the author was visiting the Simons Institute for the Theory of Computing. \href{mosheiff@bgu.ac.il}{mosheiff@bgu.ac.il}} \and Nicolas Resch\thanks{Informatics Institute, University of Amsterdam. Research supported by a Veni grant (VI.Veni.222.347) from the Dutch Research Council (NWO). Part of this work was conducted while the author was visiting the Simons Institute for the Theory of Computing.  \href{n.a.resch@uva.nl}{n.a.resch@uva.nl}} \and Kuo Shang\thanks{School of Electronic Information and Electrical Engineering, Shanghai Jiao Tong University. \href{billy63878@sjtu.edu.cn}{billy63878@sjtu.edu.cn}} \and Chen Yuan\thanks{School of Electronic Information and Electrical Engineering, Shanghai Jiao Tong University. \href{chen\_yuan@sjtu.edu.cn}{chen\_yuan@sjtu.edu.cn}}}
\date{\today}

\maketitle

\begin{abstract}
    In this work, we consider the task of generating list-decodable codes over small (say, binary) alphabets using as little randomness as possible. Specifically, we hope to generate codes achieving what we term the \emph{Elias bound}, which means that they are $(\rho,L)$-list-decodable with rate $R \geq 1-h(\rho)-O(1/L)$. A long line of work shows that uniformly random linear codes (RLCs) achieve the Elias bound: hence, we know $O(n^2)$ random bits suffice. Prior works (Guruswami and Mosheiff, FOCS 2022; Putterman and Pyne, arXiv 2023) demonstrate that just $O(Ln)$ random bits suffice, via puncturing of low-bias codes. These recent constructions are essentially combinatorial, and rely (directly or indirectly) on graph expansion.
    
    We provide two new constructions, which are \emph{algebraic}. Compared to prior works, our constructions are considerably simpler and more direct. Furthermore, our codes are designed in such a way that their \emph{duals} are also quite easy to analyze. Our first construction --- which can be seen as a generalization of the celebrated Wozencraft ensemble --- achieves the Elias bound and consumes $Ln$ random bits. Additionally, its dual code achieves the Gilbert-Varshamov bound with high probability, and both the primal and dual admit quasilinear-time encoding algorithms. The second construction consumes $2nL$ random bits and yields a code where both it \emph{and its dual} achieve the Elias bound. As we discuss, properties of a dual code are often crucial for applications of error-correcting codes in cryptography.

    In all of the above cases -- including the prior works achieving randomness complexity $O(Ln)$ -- the codes are designed to ``approximate'' RLCs. More precisely, for a given locality parameter $L$ we construct codes achieving the same $L$-local properties as RLCs. This allows one to appeal to known list-decodability results for RLCs and thereby conclude that the code approximating an RLC also achieves the Elias bound (with high probability). As a final contribution, we indicate that such a proof strategy is inherently unable to generate list-decodable codes of rate $R$ over $\F_q$ with less than $L(1-R)n\log_2(q)$ bits of randomness.
    \chen{I edit the abstract to makes the randomness bit explicit.}
    \nri{Looks good, I changed a couple things as well for wording.}
\end{abstract}

\section{Introduction} \label{sec:intro}
\input{intro}

\section{Preliminaries} \label{sec:prelims}
\input{prelims}
\section{Pseudorandom Codes from Linearized Polynomials} \label{sec:linearized-polys}
\input{linearized_polynomials}

\section{Pseudorandom Codes from Row and Column Polynomials} \label{sec:row-col}
\input{new_constructions_nic}
\section{Challenge of Moving to Sublinear Randomness} \label{sec:sublinear-randomness}
\input{sublinear_randomness}

\bibliographystyle{alpha}
\bibliography{refs}

\end{document}

%% file: intro.tex
The basic task of coding theory is to define subsets of $\cC \subseteq [q]^n$, where $q \in \N$ is the \emph{alphabet size} and $n \in \N$ is the \emph{block-length}, that satisfy two conflicting desiderata. Firstly, the code $\cC$ should be as large as possible, as this corresponds to the amount of information that one transmits in $n$ symbol transmissions. But secondly, the elements of $\cC$, termed \emph{codewords}, should be as spread out as possible in order to minimize the likelihood that two distinct codewords are confused should errors be introduced. In this work, we will focus almost exclusively on \emph{linear} codes, in which case we require $q$ to be a prime power and insist that $\cC \leq \F_q^n$, i.e., $\cC$ is a subspace of the vector space $\F_q^n$. Unless otherwise mentioned, from now on all codes are linear. 

Typically, instead of directly working with the cardinality $|\cC|$ of a code, one analyzes its rate $R = \frac{\log_q|\cC|}n = \frac{\dim(\cC)}n$, which measures the amount of information transmitted per codeword symbol. To measure a code's error-resilience, various metrics can be used. The most basic (at least for the setting of worst-case errors) is $\cC$'s minimum distance $\delta := \min\{d(\bx,\by):\bx, \by \in \cC,\bx \neq \by\}$, where here and throughout $d(\bx,\by) := \frac1n |\{i \in \{1,2,\dots,n\}: x_i \neq y_i\}|$ is the (normalized) Hamming distance. A classical observation is that as long as $\rho<\delta/2$ fraction of symbols are corrupted, one can always uniquely-decode\footnote{At least, information-theoretically. Algorithmic decoding is a separate challenge.} to recover the original codeword. 

The first question one might ask, then, is what sort of tradeoffs one can achieve between rate and distance. A classical result due to Gilbert~\cite{G52} and Varshamov~\cite{V57} states that, for any $R,\delta$ satisfying $R < 1-h_q(\delta)$,\footnote{Here, $h_q(\cdot)$ denotes the $q$-ary entropy function, which we define formally in \Cref{sec:notation}.} there exist infinite families of codes of rate at least $R$ and distance at least $\delta$. We say that codes which achieve this tradeoff (or, in some cases, get $\eps$-close for some small $\eps$) achieve the \emph{GV bound}.

A natural relaxation of unique-decoding that we focus upon is \emph{list-decoding}: for a parameter $\rho \in (0,1-1/q)$ and an integer $L \geq 1$ we call a code $\cC$ $(\rho,L)$\emph{-list-decodable} if for any $\bz \in \F_q^n$, the number of codewords at distance at most $\rho$ from $\bz$ is less than $L$. In notation:
\[
    \forall \bz \in \F_q^n, ~~ |\{\bc \in \cC: d(\bc,\bz) \leq \rho\}| < L \ .
\]
Early work due to Elias and Wozencraft~\cite{Elias57,Wozencraft58,Elias91} proposed list-decodable codes as an object of study, largely as an intermediate target on the way to unique-decoding. In the past 30 years or so, list-decodable codes have seen increased attention due to their connections to other parts of theoretical computer science, particularly complexity theory, cryptography and pseudorandomness~\cite{goldreich1989hard,babai1990bpp,lipton1990efficient,kushilevitz1993learning,jackson1997efficient,sudan2001pseudorandom}. Note that the above discussion of unique-decodability implies that any code with distance $\delta$ is $(\delta/2,1)$-list-decodable. In particular, by choosing a code $\cC$ achieving the GV bound, we can have a rate $R<1-h_q(2\rho)$ code which is $(\rho,1)$-list-decodable.

If one allows the list-size parameter $L$ to grow, the list-decoding capacity theorem essentially says that we can correct up to \emph{twice as many} errors for the same rate. More precisely, there exist $(\rho,L)$-list-decodable codes of rate $1-h_q(\rho)-O(1/L)$.\nr{Johnathan, does this last bit make sense now?} Informally, one says that a code\footnote{Technically, one should speak of an infinite family of codes of increasing block-length whose rates have limit $R$. In this work, we will not be too careful with this formalism, but it should be clear that our constructions lead to such infinite families.} \emph{achieves list-decoding capacity} if its rate is arbitrarily close to $1-h_q(\rho)$ with list-size $L \leq \poly(n)$. For our purposes, we are interested in codes that achieve the tradeoff achieved by random codes. Introducing some terminology, we will say a code construction $\cC$ achieves the \emph{Elias bound} if it is $(\rho,L)$-list-decodable and has rate at least $1-h_q(\rho)-O(1/L)$. 

We also mention that a generalization of list-decoding, termed list-recovery, has seen increasing attention in recent years. It was originally abstracted as a useful primitive in list-decoding concatenated codes~\cite{GuruswamiI01,GuruswamiI02,GuruswamiI03,GuruswamiI04}. However, it has recently proved itself to merit investigation in its own right, finding applications in cryptography~\cite{HIOS15,holmgren2021fiat}, randomness extraction~\cite{guruswami2009unbalanced}, hardness amplification~\cite{doron2020nearly}, group testing~\cite{INR10,NPR12}, streaming algorithms~\cite{doron2022high}, and beyond. 
The interested reader is directed to \Cref{sec:coding-theory} for the precise definition of list-recovery; for now, suffice it to say that all of the preceding and ensuing discussion generalizes cleanly to list-recovery as well. 

\medskip

An outstanding problem in the theory of error-correcting codes is to provide explicit\footnote{While we will not be too precise with the meaning of ``explicit'' in this work, we informally mean that a description of the code can be constructed deterministically in time polynomial in $n$.} constructions of capacity-achieving list-decodable codes. The problem in the regime of ``large alphabet'' has seen tremendous progress in the last quarter of a century. Since Guruswami and Rudra demonstrated that folded Reed-Solomon codes achieve list-decoding capacity~\cite{GR08}, a long line of work has now led to explicit constructions of capacity-achieving codes: namely, codes of rate $R$ which are $(1-R-\eps,\exp(\poly(1/\eps)))$-list-decodable, assuming $q \geq (1/\eps)^{\Omega(1/\eps^2)}$~\cite{GR21}. While achieving optimal tradeoffs between all the parameters involved is still not completely resolved, it is fair to say that we have very satisfactory constructions, assuming $q$ is sufficiently large. However, when it comes to explicitly constructing list-decodable codes over the binary alphabet, the existing results are quite paltry. The only notable successes concern the regime of very high noise, where one hopes to decode at radius $\frac12-\eps$ with codes of rate $\Omega(\eps^2)$, matching (up to constant factors) the rate-distance tradeoff achieved by random linear codes. The current state of the art is Ta-Shma's code~\cite{ta2017explicit} achieving rate $\Omega(\eps^{2+o(1)})$, for which we now additionally have efficient unique- and list-decoding~\cite{GQST20,GST21} algorithms. 

\medskip

In light of the difficulty of explicitly constructing list-decodable codes over small alphabets, we focus on a more modest goal: let's construct them randomly \emph{using as little randomness as possible}. And in this case, we would like to achieve the Elias bound, i.e., for $(\rho,L)$-list-decodability the rate $R$ should be at least $1-h_q(\rho)-O(1/L)$. For example, the classical argument of Elias -- which argues that random subsets $\cC \subseteq \{0,1\}^n$ of size $2^{Rn}$ are $(\rho,L)$-list-decodable assuming $R < 1-h_2(\rho)-1/L$ -- shows that with exponentially many random bits we can have such a code. This generalizes to $R < 1-h_q(\rho)-1/L$ for general alphabet size $q$.


If rather than a plain random code one instead samples a random linear code (RLC), a long line of works \cite{ZyablovP81,GuruswamiHSZ02,GHK11,CheraghchiGV13,Wootters13,RudraW14a,RudraW18,LiWootters,GuruswamiLMRSW20,alrabiah2023randomly} shows that they achieve the Elias bound in most parameter regimes. In particular, \cite{LiWootters,GuruswamiLMRSW20} settles the binary case. Hence, $O(n^2)$ random bits are sufficient. 

\medskip 
To push beyond this, \cite{MRRSW} shows that random low-density parity-check (LDPC) codes also achieve list-decoding capacity efficiently, and such codes can be sampled with $O(n\log n)$ bits of randomness. This work actually argues something stronger: namely, any \emph{local property} that is satisfied by a random linear code of rate $R$ with high probability is also satisfied by a random LDPC code of rate $R-o(1)$. 

While we precisely define local properties in \Cref{sec:local-props}, for now we give the following intuitive explanation: for a given locality parameter $\ell=O(1)$, $\ell$-local properties are defined by excluding a collection of ``forbidden subsets'' of size $\ell$. In the case of list-decodability, the collection would be defined as the family of all $L$-tuples of vectors $\bx_1,\dots,\bx_L$ which all lie in a Hamming ball of radius $\rho$. That is, $(\rho,L)$-list-decodability is an $L$-local property. The same in fact holds for $(\rho,\lambda,L)$-list-recoverability: it is also an $L$-local property. 

\medskip
Subsequent work by Guruswami and Mosheiff~\cite{GM22} provides a means of sampling codes achieving list-decoding capacity efficiently with only $O(n)$ randomness. In fact, as is the case for LDPC codes, these codes achieve the same local properties as RLCs. First, note that an RLC is nothing but a random puncturing of the Hadamard code.\footnote{Recall that the Hadamard code encodes a message $\bm{m} \in \F_2^k$ into a length-$2^k$ codeword by computing $\langle \bm{m}, \bx\rangle$ for every $x\in \F_2^k$.} Observe further that the Hadamard code is optimally balanced, in the sense that every non-zero codeword has weight precisely $1/2$. Gurusawmi and Mosheiff suggest then puncturing some other explicitly chosen ``mother code'' of block-length $N$, and so long as this code is nearly balanced in the sense that all non-zero codewords have weight $\approx 1/2$, then a random puncturing will again ``look like'' an RLC from the perspective of local properties. Assuming $N \leq \poly(n)$, then we need $n\log N = O(n\log n)$ random bits to sample such a code, matching the guarantee for LDPC codes. To achieve $O(n)$ randomness, one must ensure $N \leq O(n)$ (by choosing, e.g., Ta-Shma's codes~\cite{ta2017explicit} for the mother code) and then puncturing without replacement: one thus requires only $\log \binom{N}{n} = O(n)$ bits of randomness. 

\medskip

Very recently, another derandomization has been offered. Putterman and Pyne~\cite{PP23} demonstrate that instead of choosing each coordinate independently one can choose them via an expander random walk. This then means that we only require $O(n\log d)$ bits of randomness to sample the code, where $d$ is the degree of the expander graph. Assuming $d=O(1)$ -- which is achievable if one is interested in local properties -- we in particular find that $O(n)$ bits of randomness suffice. 


\medskip 

Thus, we currently know how to construct list-decodable binary codes achieving capacity efficiently with $O(n)$ bits of randomness.\nr{Removed $O_L(n)$ notation -- I don't think we need it.} As elaborated below, this seems like a hard barrier for current techniques.

The above constructions are quite ``indirect,'' requiring the existence of a sufficiently nice mother code that can then be punctured. While explicit constructions of such highly balanced codes are known, the constructions are all quite nontrivial. This status naturally leads us to wonder if we can provide more ``direct'' randomness-efficient constructions of binary codes achieving the Elias bound.

Furthermore, we also find motivations stemming from code-based cryptography. In this setting, one would often like to generate codes that ``look like'' random codes, but in fact admit efficient descriptions, as the description of the code is often some sort of public parameter that must be known by all parties making use of the cryptographic scheme. We elaborate upon this connection below. And in this case, one would often like the \emph{dual} code to also look random (again, we discuss this motivation further below). We observe that while the dual of an RLC is again an RLC -- and hence will also satisfy the Elias bound with high probability -- the above constructions (LDPC or puncturing-based) do not have such guarantees. And indeed, the dual of an LDPC code certainly cannot even have linear minimum distance! As for the puncturing-based constructions, it is unclear to us whether random puncturing can yield the Elias bound and good dual distance; at the very least, such a proof would require new techniques.


\subsection{Our Results} \label{sec:results}

In this work, we provide two new randomized constructions of codes achieving the Elias bound and consuming only $O(Ln)$ bits of randomness. In fact, for any (constant) locality parameter $\ell$, we show that these codes are $\ell$-\emph{locally similar} (see Definition \ref{def:local-sim}) to RLCs, which implies that any $\ell$-local property satisfied by RLCs with high probability is also satisfied by our codes with high probability (in fact, the success probability will be of the form $1-q^{-\Omega(n)}$). In particular, taking $\ell=L$ implies that all our codes achieve list-decoding capacity efficiently with high probability. We provide our constructions for general (but constant) field size $q$, although we are mostly motivated by the binary case. 

The notions of \emph{local property} and \emph{local similarity} are thoroughly defined and discussed in \cref{sec:local-props}. For concreteness, we give a shorter and less precise description here, and for simplicity we restrict attention here to the binary case. Fix a \emph{locality parameter} $\ell\in \N$ and consider the set of all $n \times \ell$ binary matrices. We generally think of $\ell$ as constant while $n$ tends to infinity. A code $\cC\subseteq \F_2^n$ is said to contain a matrix $A\in \F_2^{n\times \ell}$ if it contains all the columns of $A$ as codewords. We group the matrices in $\F_2^{n\times \ell}$ according to their \emph{row distribution}. More precisely, we associate with $A\in \F_2^{n\times \ell}$ a distribution $\Emp_A \sim \F_2^\ell$ that yields a vector $x\in \F_2^\ell$ proportionally to the number of times that $x$ appears as a row  in $A$, namely, $\tau(x) = \frac{|\{i\in [n] \suchthat A_i = x\}|}n$, where $A_1,\dots,A_n \in \F_2^\ell$ denote $A$'s rows. We denote the set of all matrices in $\F_2^{n \times \ell}$ with row distribution $\tau$ by $\cM_{n,\tau}$. We can now define the notion of local-similarity to an RLC for binary codes.

\begin{definition}[Local similarity to RLC in the binary case]\label{def:LocalSimilarityIntro}
     Let $\cC \leq \F_2^n$ be a linear code sampled from some ensemble. We say that $\cC$ is $\ell$\emph{-locally-similar to an RLC of rate $R$} if, for every $1 \leq b \leq \ell$ and every distribution $\tau \sim \F_2^b$ with $\dim(\tau)=b$, we have
    \[
        \Eop_{\cC}\left[|\{A \in \cM_{n,\tau}: A \subseteq \cC\}|\right] \leq 2^{(H_2(\tau)-b(1-R))n} \ .
    \]
    Above, $H_2(\tau)$ denotes the entropy of the distribution $\tau$, measured in bits.\nr{Explained $H_2(\cdot)$ notation here.} 
\end{definition}
Less formally, $\cC$ is locally similar to an RLC if, for every $\tau$, the expected number of matrices from $\cM_{n,\tau}$ in $\cC$ is not much larger than that in an RLC. The motivation for this definition is that local-similarity of $\cC$ to an RLC implies that $\cC$ almost surely satisfies every \emph{local property} (a notion formulated in \Cref{def:local-prop}) that is satisfied by an RLC  with high probability. As important motivating special cases, we note that list-decodability and list-recoverability are both local properties; this is established in e.g.~\cite{MRRSW,resch2020list}. Therefore, we can morally say that any code satisfying \Cref{def:LocalSimilarityIntro} is likely to be list-decodable and list-recoverable with similar paramters to those of an RLC. In particular, such a code is likely to achieve the Elias bound. 

\medskip

We now turn to describing our constructions. In contrast to prior works, neither of our constructions rely on an explicit ``mother code'' which we then puncture, but are instead built ``from scratch.'' Our constructions also have the pleasing property of being rather simple. A final major bonus of our codes is that their \emph{duals} also satisfy non-trivial properties: for the first construction, its dual achieves the GV bound with high probability; for the second, its dual is also $\ell$-locally similar to RLCs!

\medskip

Our first construction, which yields a code over $\F_q$ of block-length $n$, uses \emph{linearized polynomials} $f(X)$ with $q$-degree at most $\ell-1$. That is, $f(X)$ is of the form $\sum_{i=0}^{\ell-1}f_iX^{q^i}$\jm{Rev 2 complained that the fields involved are not clear in this paragraph. Let's make it clearer.} where each $f_i \in \F_{q^n}$, the degree $n$ extension of $\F_q$. As we review in \Cref{sec:alg-concepts}, such linearized polynomials define $\F_q$-linear maps $\F_{q^n} \to \F_{q^n}$. That is, for all $a,b \in \F_q$ and $\alpha,\beta \in \F_{q^n}$, we have $f(a\alpha + b\beta) = af(\alpha)+bf(\beta)$. The code is sampled by sampling the coefficients $f_i \in \F_{q^n}$ independently and uniformly at random. In particular, this requires only $\ell \lceil n \log_2q\rceil$ uniformly random bits.

To provide codes with rate $R = k/n$, we fix an $\F_q$-linear subspace $V\subseteq \F_{q^n}$ of dimension $k$. The code is then defined as $\{\varphi(f(\alpha)):\alpha \in V\}$, where $\varphi:\F_{q^n} \to \F_q^n$ is any bijective $\F_q$-linear map. Recall that such a map exists as $\F_{q^n}$ is of dimension $n$ as a vector space over $\F_q$, and any two vector spaces over the same field of the same dimension are isomorphic. For example, if $\omega_1,\dots,\omega_n$ is a basis for $\F_{q^n}$ over $\F_q$, we could set $\varphi:\sum_{i=1}^nx_i \omega_i \mapsto (x_1,\dots,x_n)$. We say that $\cC$ is a \emph{pseudorandom code from linearized polynomials} of rate $R$ and degree $\ell$, or just $\PCLP(R,\ell)$ for short, if it is sampled according to the above procedure.\footnote{The dependence on $V$ and $\varphi$ is not made explicit in this notation, as they will turn out to have no impact on our results regarding $\cC$'s combinatorial properties.} Requiring the polynomial $f$ to be linearized guarantees that the resulting code is linear, as desired. 

Not only are we able to show that such codes achieve the Elias bound with high probability, we also show that their \emph{dual code} achieves the GV bound. As we elaborate upon below, for cryptographic applications a code's dual distance is often a crucial parameter of interest. We remark that, prior to our work, we are not aware of any construction of binary codes consuming $O(n)$ randomness outputting codes with both distance and dual distance lying on the GV bound.

\medskip

Having realized that with $\ell n$ randomness we can construct a binary code that is $\ell$-locally similar to RLCs with dual code achieving the GV bound (which informally follows from being $1$-locally similar to RLCs), it is natural to wonder if it is possible to get both primal and dual code $\ell$-locally similar to RLCs. We emphasize again that this would imply that both the primal and the dual code achieve the Elias bound for list size $L=\ell$. The answer to this question is yes: our second construction has exactly this property. We now turn to describing this construction. 

First, fix distinct elements $\alpha_1,\dots,\alpha_n \in \F_{q^n}$ (they need not be linearly independent over $\F_q$). Let $\gamma:\F_{q^n} \to \F_q^n$ be a full-rank linear map (as with $\varphi$ before), and let $\eta:\F_{q^n} \to \F_q^k$ be any surjective $\F_q$-linear map. For example, if $\omega_1,\dots,\omega_n$ is a basis for $\F_{q^n}$ over $\F_q$, we could set $\gamma:\sum_{i=1}^nx_i \omega_i \mapsto (x_1,\dots,x_n)$ and $\eta:\sum_{i=1}^nx_i \omega_i \mapsto (x_1,\dots,x_k)$.

For a given rate $R=k/n$, we choose independently \emph{two} polynomials $f(X),g(X) \in \F_{q^n}[X]$ uniformly amongst all such polynomials of degree at most $\ell-1$ (unlike in the previous construction, these polynomials \emph{need not} be linearized). Note that this requires only $2\ell\lceil n\log_2q\rceil$ uniformly random bits. We then define the following two matrices $G',G'' \in \F_q^{k \times n}$: 
\begin{itemize}
    \item For each $i \in [k]$, the $i$-th row of $G'$ is defined to be $\gamma(f(\alpha_i))$.
    \item For each $j \in [n]$, the $j$-th column of $G''$ is defined to be $\eta(g(\alpha_j))$. 
\end{itemize}
We then define $G := G'+G''$ and set $\cC := \{\bx G \suchthat \bx \in \F_q^k\}$. We call a code constructed in this way a \emph{pseudorandom code from row and column polynomials} of rate $R$ and degree $\ell$, or just $\PCRCP(R,\ell)$, if it is sampled according to this procedure.\footnote{For reasons analogous to before, this notation omits mention of $\gamma$, $\eta$ and $\alpha_1,\dots,\alpha_n$.} Informally, the matrix $G'$ is responsible for ensuring that the primal code is $\ell$-locally similar to RLCs, while the matrix $G''$ guarantees the same holds for the dual code. In particular, if we had just set $G=G'$ then we would have had just the primal code $\ell$-locally similar to RLCs, while if we just set $G=G''$ then only the dual code would be $\ell$-locally similar to RLCs. 

\medskip
This second construction thus yields the following result. 

\begin{theorem} [{Informal; follows from \Cref{thm:row-column-polynomial-codes}}]
    Let $\ell,n\in \N$, $R \in (0,1)$ for which $Rn \in \N$ and $q$ is a prime power. Let $\cP,\cP^\perp$ be $\ell$-local properties, and suppose that an $\RLC(R)$ satisfies $\cP$ with probability $1-q^{-\Omega(n)}$ and an $\RLC(1-R)$ satisfies $\cP^\perp$ with probability $1-q^{-\Omega(n)}$. Then, for sufficiently large $n$, there exists a randomized procedure consuming $O(\ell n\log q)$ bits of randomness outputting a code $\cC$ such that, with probability at least $1-q^{-\Omega(n)}$:
    \begin{itemize}
        \item $\cC$ has rate $R$;
        \item $\cC$ satisfies $\cP$; and
        \item $\cC^\perp$ satisfies $\cP^\perp$, where $\cC^\perp$ is the code dual to $\cC$.
    \end{itemize}
\end{theorem}

As mentioned above, we do know of other code ensembles sampled with linear randomness that share local properties with RLCs. However, we are not aware of any other code ensembles for which the \emph{dual} code also shares local properties with RLCs. 

\medskip
Lastly, we mention one other pleasing feature of our first construction based on linearized polynomials. Namely, a careful choice of representation for $\F_{q^n}$ over $\F_q$ allows one to view the task of encoding a message as a constant number of polynomial multiplications, which can be computed in $O(n\log n)$ time via standard FFT-type methods. Thus, we can also claim the following result.

\begin{theorem} [{Informal; follows from \Cref{thm:PCLP-main} and \Cref{prop:PCLP-sampling-and-encoding}}]
    Let $\ell,n \in \N$, $R \in (0,1)$ for which $Rn \in \N$ and $q$ a prime power. Let $\cP$ be an $\ell$-local property, and suppose that an $\RLC(R)$ satisfies $\cP$ with probability $1-q^{-\Omega(n)}$. Then, for sufficiently large $n$, there exists a randomized procedure consuming $O(\ell n\log q)$ bits of randomness outputting a code $\cC$ such that, with probability at least $1-q^{-\Omega(n)}$:
    \begin{itemize}
        \item $\cC$ has rate $R$;
        \item $\cC$ satisfies $\cP$;
        \item $\cC^\perp$ achieves the GV-bound;
        \item $\cC$ is encodable in $O(n\log n)$ time.
    \end{itemize}
\end{theorem}




\paragraph{Cryptographic motivation.}\jm{Rev 3 says that the crypto motivation here is pretty vague. I tend to agree. Can we say something more concrete?} We remark that codes that ``look like random linear codes'' but are in fact samplable with less randomness are highly motivated by cryptographic considerations. And in fact, achieving good dual distance is often a crucial desideratum: the security of a cryptosystem is typically tied to the dual distance of the code, whether this is provably the case (i.e., with secret-sharing schemes \cite{CDDFS15},\cite{CXY20}) or plausibly the case (i.e., the linear tests framework for learning parity with noise~\cite{BCGIKRS22}). However, codes that require less randomness to generate allow for reduced public key sizes: the sizes of keys is typically the major drawback of post-quantum public-key cryptosystems. Hence the popularity of, e.g., quasi-cyclic~\cite{HQC} and moderate-density parity-check codes~\cite{BIKE}.

While for McEliece-type encryption schemes~\cite{M78} an important requirement is that the code admits an efficient decoding algorithm, this is not in fact required for recent applications of error-correcting codes in multi-party computation -- e.g., in the context of pseudorandom correlation generators (PCGs). In fact, one typically hopes that such codes \emph{do not} admit efficient decoding~\cite{BIPSW18,DGHIKSZ21}. A current ``rule-of-thumb'' is that the employed code should have good dual distance. In our view, a much more satisfying guarantee is that the dual code in fact shares more sophisticated properties with random linear codes, e.g., list-decodability/-recoverability, as our techniques can establish. 

We further remark that many recent code constructions for PCGs~\cite{boyle2020correlated,BCGIKRS22,raghuraman2023expand} in fact only admit randomized constructions that fail with probability $1/\poly(n)$; that is, they fail with non-negligible probability. This implies that the resulting constructions technically fail to satisfy standard security definitions. In contrast, all of our code constructions satisfy the targeted combinatorial properties with probability at least $1-\exp(-\Omega(n))$.

\nri{Tried to add something more concrete here. Is it helpful or just confusing?} 
\jmi{I think it's good}

Concretely, one can plug the (dual of) our first code construction into the framework of~\cite{boyle2019efficient} to obtain PCGs for standard correlations like oblivious transfers with quasi-linear computation time for the involved parties. While constructions of such efficiency were known previously, we view our additional guarantee of local-similarity to RLCs as a stronger security guarantee than prior constructions offered (which only promised good minimum distance). Additionally, as emphasized above, ours is the first construction of such efficiency with negligible failure probability (in the randomized construction of the utilized code). We leave further investigation of the PCG implications of our codes for future research. 

\medskip 
Finally, we recall that a linear code with large distance and dual distance yields a linear secret sharing scheme with small reconstruction and large privacy, and moreover that an asymptotically good linear code with asymptotically good dual yields an asymptotically good linear secret sharing scheme. The asymptotic linear secret sharing scheme was first considered and realized in~\cite{CC06}, thereby enabling an ``asymptotic version'' of the general MPC theorem from~\cite{BGW88}. Since 2007, with the advent of the so-called ``MPC-in-the-head paradigm''~\cite{IKOS09}, these asymptotically-good schemes have been further exposed as a central theoretical primitive in numerous constant communication-rate results in multi-party cryptographic scenarios and -- perhaps surprisingly -- in two-party cryptography as well. Druk and Ishai~\cite{DI14} utilize an expander graph to construct a linear time encodable code; such a code combined with a linear-time universal hash function \cite{CDDFS15} yields an asymptotically good linear secret sharing scheme equipped with a linear time encoding algorithm. Recently, Cramer, Xing and Yuan~\cite{CXY20} construct an asymptotically good secret sharing scheme with quasi-linear time encoding and decoding algorithm. 

We remark that the privacy and reconstruction of all above mentioned asymptotically good schemes do not achieve the optimal trade-off, i.e., GV bound. In contrast, the linear code derived from our linearized polynomials yields an asymptotically good linear secret sharing scheme with quasi-linear-time encoding algorithm, and moreover the privacy and reconstruction of the resulting scheme achieves the optimal trade-off. 
\nri{@Chen, do you think you could state something a bit more formal about what our codes allow us to achieve?}

\medskip

\paragraph{Challenge of sublinear randomness.} As a final contribution, we highlight the inherent challenge of designing code ensembles consuming $o(n)$ random bits outputting codes that achieve the EB bound with high probability -- or for that matter, even the GV-bound.\footnote{Of course, recent breakthroughs~\cite{ta2017explicit} provide \emph{explicit} binary codes of rate nearly $\Omega(\eps^2)$ with minimum distance $\frac12-\eps$; however such constructions seem inherently unable to reach the GV bound in other regimes, and even in the $\frac12-\eps$ distance regime the constant in front of the rate is unlikely to be pushed to $\frac{2}{\ln 2}$, as one would hope.} More precisely, we observe that any code ensemble that is $\ell$-\emph{locally similar}
to RLCs requires at least $\ell(1-R)n\log_2q$ bits of randomness. This is not much more than an observation -- namely, that the granularity required by certain distributions is only achievable with this many bits of randomness -- but we nonetheless believe that elucidating this shortcoming is insightful. Note that local similarity to RLC is merely a sufficient condition for a code ensemble to share combinatorial properties with RLCs. However, we emphasize that all the previous works (including our own) rely on local-similarity. 

We also observe that our lower bound is tight: a simple twist on our codes from linearized polynomials $\PCLP$ requires only $\ell (1-R)n\log_2 q$ to sample and is $\ell$-locally similar to RLCs. In fact, this code is a natural generalization of the famous Wozencraft ensemble~\cite{M63}. \jm{Rev 3 wanted us to say more about the resemblence to Wozencraft}\nr{Added details to flesh this out earlier.} That is, recall that the Wozencraft ensmeble is obtained by uniformly sampling $\beta \in \F_{q^k}$ and then defining
\[
    \cC := \{(\varphi(\alpha),\varphi(\beta\alpha)) : \alpha \in \F_{q^k}\} \ ,
\]
where $\varphi:\F_{q^k} \to \F_q^k$ is a full-rank $\F_q$-linear map. Defining $f(X) = \beta X$, the codewords of $\cC$ are thus of the form $(\varphi(\alpha),\varphi(f(\alpha)))$. Note that $f(X)$ is in fact a uniformly random linearized polynomial of $q$-degree at most $0$. The generalization that we consider is thus to allow $f(X)$ to have $q$-degree at most $\ell-1$, and we observe that indeed this code ensemble will be $\ell$-locally similar to RLCs. However, a drawback is that this construction only naturally produces codes of rate $1/2$; by sampling $r$ independent linearized polynomials we can also achieve rates of the form $1/r$ for $r \in \N$, but not any possible rate as we can with $\PCLP$'s (which can themselves be similarly considered a different generalization of the Wozencraft ensemble). Further discussion of this construction is provided in \Cref{subsec:generalized-wozencraft-ensemble}.

\medskip

To conclude this discussion, we provide \Cref{fig:table-of-results} summarizing our contributions and the prior state-of-the-art. \nr{I thought this belonged at the end. Also didn't want it to float so much, hence the [H].}

\begin{figure}[H] 
    \centering
    \begin{tabular}{|c|c|c|c|}
        \hline
        Source & Code & Randomness & Dual Code\\
        \hline
        \cite{GuruswamiLMRSW20} & Random Linear Code & $O(n^2)$ & EB\\
        \hline
        \cite{MRRSW} & Low-Density Parity-Check Codes & $O(Ln\log n)$ & \xmark \\
        \hline
        \cite{GM22} & Puncturing of Low-Bias Code & $O(Ln)$ & \xmark \\
        \hline
        \cite{PP23} & Expander-Puncturing of Low-Bias Code & $O(Ln)$ & \xmark\\
        \hline
        \Cref{sec:linearized-polys} & Codes from Linearized Polynomials & $Ln$ & GV\\
        \hline
        \Cref{sec:linearized-polys} & Generalized Wozencraft Ensemble & $L(1-R)n$, $R=\frac{1}{\text{integer}}$ & \xmark \\
        \hline
        \Cref{sec:row-col} & Row-Column Polynomial Codes & $2Ln$ & EB \\
        \hline
          \Cref{sec:sublinear-randomness} & Lower Bound for RLC-similarity& $L(1-R-\eps)n$ &  \\
        \hline
    \end{tabular}
    \caption{Randomness requirements for binary codes achieving the Elias Bound. We remark that all the above constructions generalize to larger (but constant) $q$. Regarding the dual code criterion, ``EB'' means that the dual-code also achieves the Elias Bound (for lists of size $L$), while ``GV'' means that the dual distance achieves the GV bound. An \xmark~means that no guarantees are provided (and, in certain cases, cannot hold). The lower bound applies to all ensembles that achieve similarity to RLC (a stronger property than the Elias bound; see Definitions \ref{def:LocalSimilarityIntro}, \ref{def:local-sim}), including all constructions listed in this table.}
    \label{fig:table-of-results}
\end{figure}

\subsection{Techniques} \label{sec:techniques}

Given a random code $\cC$ of (design) rate $R$ sampled according to either of the above constructions, we wish to demonstrate that it behaves combinatorially much as an RLC $\cD$ of rate $R$. More precisely, we consider any property $\cP$ obtained by forbidding $\ell$-sized sets of vectors and wish to show that $\cC$ satisfies the property $\cP$ with probability roughly the same as $\cD$. As discussed above, these properties capture well-studied notions like list-decodability and list-recoverability. 

Fortunately, recent works~\cite{MRRSW,GM22} have introduced a calculus for making such arguments. Intuitively, a conclusion of these works is that it suffices to argue that, for any $S \subseteq \F_q^n$ of size $\ell$, the probability that $S$ is contained in $\cC$ is roughly the same as the probability this holds for $\cD$. Of course, this latter probability is $q^{-(1-R)n\rank(S)}$, where $\rank(S)$ denotes the dimension of the vector space spanned by $S$.\footnote{At least, this holds exactly if one samples a RLC by choosing a uniformly random parity-check matrix, which is the model we consider in this work. For other natural models -- e.g., sampling a uniformly random generator matrix -- $q^{-(1-R)n\rank(S)}$ gives an upper on this probability.}

In a bit more detail, these works in fact view such sets as matrices in $\F_q^{n \times \ell}$, and observe that the forbidden matrices for properties like list-decoding are closed under row-permutation. One can thus restrict to the various orbit classes of this action, and study these orbit classes one at a time. The requirement is in fact that, for each orbit class, the expected number of matrices from that class lying in $\cC$ is roughly the same as the expected number lying in $\cD$. By identifying these orbit classes with row distributions, one obtains \Cref{def:LocalSimilarityIntro}.\nr{Added a pointer back to the previous definition.}


For our specific constructions, for fixed vectors $\bx \in \F_q^n$ we consider event indicator random variables $X_{\bx}$ outputting $1$ if $\bx \in \cC$, and observe that, for any $1 \leq b \leq \ell$, a $b$-tuple of random variables $(X_{\bx_1},\dots,X_{\bx_b})$ is independent if $\bx_1,\dots,\bx_b$ are linearly independent. Of course, this also holds for random linear codes (in fact, it holds for tuples of all sizes),\footnote{At least, this is true if one defines a random linear code by sampling a uniformly random parity-check matrix, which we here implicitly assume,} and this is the sense in which our constructions approximate the ``local behaviour'' of random linear codes, which we can then bootstrap into full-blown $\ell$-local similarity via the machinery of~\cite{MRRSW,GM22}. 

\medskip

To analyze our first construction based on linearized polynomials, we exploit the fact that for any fixed tuple of inputs and outputs $(x_1,y_1),\dots,(x_b,y_b)$ with $x_1,\dots,x_b \in \F_{q^n}$ linearly independent over $\F_q$, over a uniformly random choice of linearized polynomial $f(X)$ of $q$-degree at most $\ell-1$, the vector $(f(x_1),\dots,f(x_b))$ is distributed uniformly at random over $\F_{q^n}^b$. This follows readily from properties of the Moore matrix $M=(\alpha_i^{q^j})_{ij}$ (recalling $b \leq \ell$).

Next, we consider the code's dual. In order to show that the dual of a $\PCLP(R,\ell)$ code achieves the GV-bound with high probably, we exploit a pleasant characterization of its dual. Namely, the dual is of the form $\{\psi(f_0^{-1} \cdot \beta):\beta \in W\}$, where $W$ is connected to the dual of $V$ and $\psi$ is (morally) dual to $\varphi$.\footnote{More precisely, if $\{\alpha_1,\dots,\alpha_n\}$ is a basis for $\F_{q^n}$ for which $\varphi\left(\sum_{i}x_i\alpha_i\right) = (x_1,\dots,x_n)$, then $\psi\left(\sum_i y_i \beta_i\right) = (y_1,\dots,y_n)$ where $\{\beta_1,\dots,\beta_n\}$ is the dual basis.} In particular, the dual is essentially another $\PCLP(1-R,1)$-code! Hence, the previous discussion implies it is $1$-locally similar to an RLC, which means in particular that it achieves the GV bound.

\medskip

For the second construction, i.e., pseudorandom codes from row and column polynomials, upon summing over all choices of (full-rank) sets of message vectors one can observe that the desired behaviour of the random variables $X_{\bx_i}$ in fact follows from the following: if $X \in \F_q^{b \times k}$ is of rank $b$, then over the randomness of the generator $G$, $XG \in \F_q^{b \times n}$ is uniformly random. Recalling $G = G'+G''$, we just must show $XG'$ is uniformly random. Exploiting the requirement that $\gamma:\F_{q^n} \to \F_q^n$ is an isomorphism, it suffices to show that the tuple
\begin{align} \label{eq:informal-ell-independence}
    \left(\sum_{i=1}^k X_{ji} f(\alpha_i)\right)_{j \in [b]} \in \F_{q^n}^b 
\end{align}
is uniformly random over the choice of $f$. And this follows naturally from properties of the Vandermonde matrix, as the $\alpha_i$'s are distinct and $f$ is chosen uniformly amongst all polynomials of degree $\leq \ell-1$ (recalling again $b \leq \ell$).

The argument establishing $\ell$-local-similarity for the dual is almost identical to the above argument for the primal. Here, it suffices to consider a matrix $X \in \F_q^{n \times b}$ of rank $b$ and argue that over the randomness of $G''$ now, $G''X$ is uniformly random. And to do this, one reduces to studying a tuple of random variables analogous to those in \Cref{eq:informal-ell-independence}, although in this case the polynomial $g(X)$ will play the starring role. Since $g(X)$ is again uniformly sampled from all polynomials of degree at most $\ell-1$, the desired conclusion follows.


%% file: prelims.tex
\subsection{Miscellaneous Notation}  \label{sec:notation}
By default, $\N=\{1,2,\dots,\}$, i.e., $0\notin \N$. For a positive integer $n \in \mathbb N$, $[n]:=\{1,\dots,n\}$. Throughout, $q$ denotes a prime power, $\F_q$ denotes a finite field with $q$ elements, and $\F_{q^n}$ denotes a degree $n$ extension of $\F_q$ (which of course has size $q^n$). The \emph{$q$-ary entropy function} is defined for $x \in (0,1)$ as
\begin{align*}
    h_q(x) := x\log_q(q-1) + x\log_q\frac1x + (1-x)\log_q\frac1{1-x} \ 
\end{align*}
and extended by continuity to $h_q(0)=0$ and $h_q(1)=\log_q(q-1)$. This function is known to be monotonically increasing from $0$ to $1$ on the interval $[0,1-1/q]$, and hence we can define its inverse $h_q^{-1}:[0,1] \to [0,1-1/q]$. 

Given a discrete distribution $\tau$ and a universe $U$, we write $\tau \sim U$ to denote that $\tau$ is distributed over $U$, i.e., that $\tau$'s support $\supp(\tau) := \{x : \tau(x)>0\} \subseteq U$. In general, we write vectors in \textbf{boldface} -- e.g., $\bx$, $\by$, etc. -- while scalars are unbolded.

\subsection{Algebraic Concepts} \label{sec:alg-concepts}

Let $\Tr: \F_{q^n}\rightarrow \F_q$ be a trace function, i.e., 
$$
\Tr(\alpha)=\sum_{i=0}^{n-1}\alpha^{q^i}.
$$
\begin{lemma}\label{lm:dualbasis}
Suppose $\alpha_1,\ldots,\alpha_n$ is a basis of $\F_{q^n}$ over $\F_q$. We can always find a dual basis $\beta_1,\ldots,\beta_n$ in $\F_{q^n}$ such that 
\begin{equation}\label{eqn:dualbasis}
\Tr(\alpha_i\beta_j)=\delta_{ij}
\end{equation}
where $\delta_{ij}=0$ for any $i\neq j$ and is otherwise $1$. 
\end{lemma}
\begin{proof}
	We provide the proof for completeness. Write $\beta_i=\sum_{r=1}^n b_{i,r}\alpha_r$ and we consider the equations 
	$$
	\delta_{ji} = \Tr(\alpha_j\beta_i) = \Tr\left(\alpha_j\sum_{r=1}^n b_{i,r}\alpha_r \right) \ .
	$$
	Define the $n\times n$ matrix $T=(\Tr(\alpha_j\alpha_r))_{j, r \in [n]}$ over $\F_q$; the above $n$ equations can be written as $T(b_{i,1},\ldots,b_{i,n})^\top=\be_i$ where $\be_i$ is the $i$-th vector in the standard basis of $\F_q^n$. Since $\alpha_1,\ldots,\alpha_n$ forms a basis of $\F_{q^n}$, $T$ has full rank and there must exists a nonzero solution for $b_{i,1},\ldots,b_{i,n}$. Thus, we can always find $\beta_1,\ldots,\beta_n$ which satisfy \eqref{eqn:dualbasis}. It remains to show that $\beta_1,\ldots,\beta_n$ are $\F_q$-linearly independent. Assume not and without loss of generality we may assume $\beta_n$ can be represented as the linear combination of $\beta_1,\ldots,\beta_{n-1}$, i.e., $\beta_n=\sum_{i=1}^{n-1}\lambda_i \beta_i$ with $\lambda_1,\ldots,\lambda_{n-1}\in \F_q$. From \eqref{eqn:dualbasis} we have 
	$$
	1=\Tr(\alpha_n \beta_n)=\sum_{i=1}^{n-1}\lambda_i \Tr(\alpha_n\beta_i)=0 \ ,
	$$
	a clear contradiction.
\end{proof}

We now introduce the concept of \emph{orthogonality} for two $\F_q$-subspaces of $\F_{q^n}$. 
\begin{definition} \label{def:dual-space}
	Let $V,W\subseteq \F_{q^n}$ be a $\F_q$-linear space. $W$ is said to be \emph{orthogonal} to $V$ if the following holds
	$$\Tr(ab)=0, \quad \forall a\in V, b\in W$$
	We write $W\perp V$ to denote that $W$ is orthogonal to $V$. 
	If $\dim(W)+\dim(V)=n$, $W$ is said to be the \emph{dual space} of $V$. 
\end{definition}

Finally, we collect terminology connected to linearized polynomials. 

\begin{definition} [Linearized Polynomial]
    We call a polynomial $f(X) \in \F_{q^n}[X]$ a \emph{linearized polynomial} if it is of the form $\sum_{i=0}^{d}f_iX^{q^i}$ with $f_i \in \F_{q^n}$ and $d \in \N$. That is, the only monomials appearing in $f(X)$ have exponent a power of $q$. Its \emph{$q$-degree} is $\max\{i : f_i \neq 0\}$. 
\end{definition}

Recall that the Frobenius automorphism $\mathrm{Frob}:\F_{q^n} \to \F_{q^n}$ defined by $\mathrm{Frob}(\alpha) = \alpha^q$ is \emph{$\F_q$-linear}, i.e., it holds that for all $a,b \in \F_q$ and $\alpha,\beta \in \F_{q^n}$, 
\[
    (a\alpha + b\beta)^q = a\alpha^q + b\beta^q \ . 
\]
This readily implies that any linearized polynomial is also an $\F_q$-linear map, justifying the name. We record this fact now.

\begin{proposition} \label{prop:linearized-polys-linear}
    Any linearized polynomial defines an $\F_q$-linear map from $\F_{q^n} \to \F_{q^n}$.
\end{proposition}

\subsection{Coding Theory} \label{sec:coding-theory}

A \emph{linear code} $\cC$ is a subspace of $\F_q^n$ for a prime power $q$. Such a code may always be presented in terms of its \emph{generator matrix}, which is a matrix $G \in \F_q^{k \times n}$ for which $\cC = \{\bm{m}G:\bm{m} \in \F_q^k\}$. When $q=2$, a code is called \emph{binary}. The \emph{block-length} of the code is $n$ and its \emph{rate} is $R:=\frac kn$, where $k = \dim(\cC)$. We endow $\F_q^n$ with the (relative) \emph{Hamming metric} $d(\bx, \by):=\tfrac1n|\{i \in [n]:x_i \neq y_i\}|$ for $\bx,\by \in \F_q^n$. For a linear code $\cC \leq \F_q^n$, its {dual code} is defined as $\cC^\perp:=\{\by \in \F_q^n:\forall \bx \in \cC,~\langle\bx,\by\rangle=0\}$. \footnote{Note the contrast with \Cref{def:dual-space}: that definition is concerned with $\F_q$-linear subspaces of the ambient space $\F_{q^n}$, which is endowed with a different inner-product than $\F_q^n$. The appropriate meaning of $\perp$ can therefore be deduced from the context.} 


A \emph{random linear code} $\cC$ of rate $R=k/n$ -- briefly, a RLC$(R)$ -- is defined to be the kernel of a uniformly random matrix $H \in \F_q^{(n-k) \times n}$, i.e., $\cC := \{\bx \in \F_q^n : H\bx^\top=0\}$. 

This work concerns combinatorial properties of linear codes. The quintessential example of such a property is \emph{minimum distance} defined as $\delta = \min\{d(\bx,\by:\bx \neq \by, \bx,\by \in \cC\}$. Equivalently, it is $\min\{\wt(\bx):\bx \in \cC\setminus \{0\}\}$, the minimum weight of a non-zero codeword. By the triangle-inequality for the Hamming metric, it is immediate that $\delta/2$ is the maximum radius at which one can hope to uniquely-decode from worst-case errors. If one relaxes the requirement for unique-decoding and is satisfied with outputting a list of possible messages, then one arrives \emph{list-decoding}. 

\begin{definition} [List-Decodability]
    Let $\rho \in (0,1-1/q)$ and $L \geq 1$. A code $\cC \subseteq \F_q^n$ is $(\rho,L)$\emph{-list-decodable} if for all $\bz \in \F_q^n$, 
    \[
        |\{\bc \in \cC: d(\bc,\bz) \leq \rho\}| < L \ .
    \]
\end{definition}

A generalization of list-decoding is proferred by list-recovery. For this notion, we extend the definition of Hamming distance to allow one of the arguments to be a tuple of sets $\bS = (S_1,\dots,S_n)$, where each $S_i \subseteq \F_q$, as follows: $d(\bx,\bS) := \frac1n|\{i \in [n]:x_i \notin S_i\}|$. 

\begin{definition} [List-Recovery]
    Let $\rho \in (0,1-1/q)$, $1\leq \lambda \leq q-1$ and $L \geq 1$. A code $\cC \subseteq \F_q^n$ is $(\rho,\lambda,L)$\emph{-list-recoverable} if for all tuples $\bS = (S_1,\dots,S_n)$ with each $S_i \subseteq \F_q$ satisfying $|S_i|\leq \lambda$, 
    \[
        |\{\bc \in \cC: d(\bc,\bS) \leq \rho\}| < L \ .
    \]
\end{definition}


These are both special cases of the much more general class of \emph{local properties}, which we now introduce. The technical terminology takes some time to motivate and define, but will allow for a very clean argument once we have it in place. 

\subsection{Local Properties} \label{sec:local-props}

We now introduce the specialized notations and tools that we need in order to apply the machinery of~\cite{MRRSW,GM22}. Generally speaking, this machinery allows us to efficiently reason about the probability that sets of $\ell$ vectors (for any integer $\ell = O(1)$) lie in random ensembles of codes. In fact, it is convenient to (arbitrarily) order these sets and thereby view them as matrices. Thus, for $A \in \F_q^{n \times \ell}$ we will talk about events of the form ``$A \subseteq \cC$'', which denotes that event that every column of $A$ is contained in $\cC$.

To index these events, we assign to each matrix a \emph{type}, which is determined by the \emph{empirical row distribution} of the matrix.

\begin{definition} [Empirical Row Distribution]
    Let $A \in \F_q^{n \times \ell}$. We define its \emph{empirical row distribution} $\Emp_A\sim \F_q^\ell$ as
    \[
        \Pr_{i \in [n]}[A_i = x] \ ,
    \]
    where $A_i$ denotes the $i$-th row of $A$. 
\end{definition}

\begin{definition} [Matrix Type]
    Let $\ell \in \N$ and let $\tau \sim \F_q^\ell$. For $n \in \N$, we denote
    \[
        \cM_{n,\tau} := \left\{A \in \F_q^{n \times \ell} : \Emp_A=\tau\right\} \ .
    \]
\end{definition}

For a distribution $\tau \sim \F_q^\ell$, we denote by $\dim(\tau) = \dim(\mathrm{span}(\supp(\tau)))$. If $\dim(\tau)=\ell$, we say that $\tau$ is \emph{full-rank}. Given a linear map $B:\F_q^\ell \to \F_q^b$ we denote by $B\tau$ the distribution of the random vector obtained by first sampling $\bv \sim\tau$ and subsequently outputting $B\bv$. 

In order for a set $\cM_{n,\tau}$ to be non-empty, note that it is necessary that for all $\bv \in \F_q^\ell$, $n \cdot \tau(\bv) \in \{0,1,\dots,n\}$. This implies that there are about $(n+1)^{q^\ell}$ choices for $\tau$ for which $\cM_{n,\tau} \neq \emptyset$; for our regime of parameters, this is polynomial in the growing parameter $n$ and therefore negligible.

Furthermore, if $\cM_{n,\tau} \neq \emptyset$ we can rely on standard estimates on multinomial coefficients (e.g., \cite[Lemma 2.2]{CS04}) to conclude
\begin{align} \label{eq:estimate-on-type-class}
        n^{-O(q^\ell)}\cdot q^{n H_q(\tau)} \leq |\cM_{n,\tau}|\leq q^{n \cdot H_q(\tau)}
\end{align}
where $H_q(\tau) = -\sum_{\bv} \tau(\bv)\log_q\tau(\bv)$ is the \emph{base-$q$ entropy} of $\tau$. 
    
In particular, these estimates show that in order to study any property defined by forbidding sets of types (which, as discussed in e.g.~\cite{MRRSW,resch2020list}, includes list-decoding/-recovery) for random linear codes, it suffices to just consider one type $\tau$ at a time (as the failure bounds are always of the form $q^{-\Omega(n)}$, allowing for a union bound over the at most $\poly(n)$ forbidden types). 

We now formally define a local-property, which is parametrized by an integer $\ell$ that we think of as a constant. A \emph{property} of $q$-ary codes is simply a family of $q$-ary codes $\cC \leq \F_q^n$.\footnote{Of course, one could consider properties of general, not necessarily linear, codes, but we do not need to do that here.} We will restrict attention to \emph{local properties}, 
which we define precisely below, but informally are properties that are defined by only including codes that do not contain certain ``forbidden sets,'' with the additional proviso that these forbidden sets all have constant size. Note that this includes properties like $(\rho,L)$-list-decoding: such a property is defined by forbidding all $L$-sized subsets of $\F_q^n$ for which all the vectors lie in some Hamming ball of radius $\rho$, i.e., subsets $\{\bx_1,\dots,\bx_{L}\}$ such that there exists $\bz \in \F_q^n$ for which $\max_{i \in [L]} d(\bx_i,\bz) \leq \rho$ (see \cite[Sec.\ 2.5]{GMRSW22} for a formal discussion of why \emph{list-recoverability} -- which we recall generalizes list-decodability -- is a local property in the sense of Definition \ref{def:local-prop}).

\begin{definition} [Local Property] \label{def:local-prop}
Let $\ell,n \in \N$ and $q$ a prime power. A property $\cP$ of block-length $n$ $q$-ary codes is called an $\ell$-\emph{local property} if there exists a finite set $T$ of distributions $\tau \sim \F_q^\ell$ such that 
\[
    \cC \in \cP \iff \forall \tau \in T, \forall A \in \cM_{n,\tau},~A \not\subseteq \cC \ .
\]
That is, such a property is defined by forbidding any matrices of some forbidden type. 
\end{definition}


\begin{remark} \label{rem:row-symmetry}
    Technically, what we have defined would be referred to as a local, \emph{row-symmetric} and \emph{monotone-decreasing} property in other works \cite{GMRSW22,GM22}. In this work all considered properties are row-symmetric and monotone-decreasing, and we omit these additional adjectives.
\end{remark}


As a final concept surrounding local properties, we define \emph{threshold rates} for RLCs, which intuitively describe the maximum rate of a RLC so that we can hope the RLC will satisfy the property $\cP$. 

\begin{definition} [Threshold Rate for RLCs]
    Let $\cP$ be an $\ell$-local property of block-length $n$ codes in $\F_q^n$. We define 
    \[
        \TRLC(\cP) := \sup\{R \in [0,1]:\Pr[\RLC(R) \text{ satisfies } \cP] \geq 1/2\} \ .
    \]
\end{definition}

A crucial property of $\TRLC(\cP)$ is the following ``sharp threshold''-type phenomenon.

\begin{proposition} [{\cite[Lemma 2.7]{MRRSW}}] \label{prop:sharp-threshold}
    Let $\cP$ be an $\ell$-local property of block-length $n$ codes in $\F_q^n$. Let $\eps>0$. 
    \begin{itemize}
        \item If $R \leq \TRLC(\cP)-\eps$ then $\Pr[\RLC(R) \text{ satisfies } \cP] \geq 1-q^{-\eps n}$.
        \item If $R \geq \TRLC(\cP)+\eps$ then $\Pr[\RLC(R) \text{ satisfies } \cP] \leq \binom{n+q^{2\ell}-1}{q^{2\ell}-1}^3\cdot q^{-\eps n}$.
    \end{itemize}
\end{proposition}

\begin{remark}
    We remark that the above proposition follows from a characterization of $\TRLC(\cP)$, describing it as a sort of optimization problem. And in fact, the proof of \Cref{prop:GM22} below crucially uses this characterization. However, we do not need this characterization for our results, and hence refrain from providing it.
\end{remark}

Now, recall that we are not in fact directly interested in understanding combinatorial properties of (uniformly) random linear codes (RLCs). Rather, we define other ensembles of linear codes, and would like to argue that they inherit properties of RLCs. More precisely, for any $\ell$-local property $\cP$, we would like to argue that so long as we sample codes of rate $\TRLC(\cP)-\eps$, then our codes will also satisfy $\cP$ with high probability. To do this, we use the concept of \emph{local similarity}, which is inspired by concepts arising in~\cite{GM22}.


\begin{definition} [Local Similarity] \label{def:local-sim}
    Let $\cC \leq \F_q^n$ be a linear code sampled from some ensemble. We call $\cC$ $\ell$\emph{-locally-similar} to $\RLC(R)$ if, for every $1 \leq b \leq \ell$ and every distribution $\tau \sim \F_q^b$ with $\dim(\tau)=b$, we have
    \[
        \Eop_{\cC}\left[|\{A \in \cM_{n,\tau}: A \subseteq \cC\}|\right] \leq q^{(H_q(\tau)-b(1-R))n} \ .
    \]
\end{definition}

The utility of local similarity stems from the following result which states that if a property holds with high probability for RLCs of rate $R$, then it also holds for any ensemble that is locally-similar to RLCs of rate $R$. 

\begin{proposition} [{\cite[Lemma 6.12]{GM22}}] \label{prop:GM22}
    Let $n \in \N$, $q$ a prime power and $\ell \in \N$ such that $\frac{n}{\log_qn} \geq \omega_{n \to \infty}(q^{2\ell})$. Let $\cC$ be a code sampled from an ensemble that is $\ell$-locally-similar to $\RLC(R)$.

    Let also $\eps>0$. Then, for any row-symmetric and $\ell$-local property $\cP$ over $\F_q^n$ such that $R \leq \TRLC(\cP)-\eps$, it holds that 
    \[
        \Pr_{\cC}[\cC \text{ does not satisfy }\cP] \leq q^{-n(\eps - o_{n \to \infty}(1))} \ .
    \]
\end{proposition}



For both our codes, we will in fact prove that for any matrix $A \in \F_q^{n \times b}$ of rank $b$, we have $\Pr[A \subseteq \cC] = q^{-n(1-R)b}$. We demonstrate below that this readily implies local similarity. (This is also implicit in~\cite{GM22}, but we spell it out for clarity, and as it is a fact that we will use twice.)

\begin{proposition} \label{prop:prob-bound-loc-sim}
    Let $n \in \N$, $q$ a prime power and $\ell \in \N$ such that $\frac{n}{\log_qn} \geq \omega_{n \to \infty}(q^{2\ell})$. Let $\cC \leq \F_q^n$ be a linear code sampled from some ensemble. Suppose that for any $A \in \F_q^{n \times b}$ of rank $b$, we have $\Pr[A \subseteq \cC]\leq q^{-n(1-R)b}$ for some $R\in [0,1]$. Then $\cC$ is $\ell$-locally-similar to $\RLC(R)$.

    In particular, for any $\ell$-local property $\cP$ over $\F_q^n$ such that $R \leq \TRLC(\cP)-\eps$, 
    \[
        \Pr_{\cC}[\cC \text{ does not satisfy }\cP] \leq q^{-n(\eps - o_{n \to \infty}(1))} \ .
    \]
\end{proposition}

\begin{proof}
    For any $1\leq b\leq \ell$ and every distribution $\tau\sim \F_q^b$, we must show that 
	 $$
	   \Eop_{\cC}\left[|\{A \in \cM_{n,\tau}: A \subseteq \cC\}|\right] \leq q^{(H_q(\tau)-b(1-R))n} \ .
	$$
	By assumption, for any $A\in \cM_{n,\tau}$,
	$$
	   \Pr[A\subseteq \cC]=q^{-(1-R)\rank(A)} \ .
	$$
	We can thus take a union bound over all matrices in $\cM_{n,B\tau}$, recalling from~\Cref{eq:estimate-on-type-class} that
    \[
        |\cM_{n,\tau}|\leq q^{n \cdot H_q(\tau)} \ , 
    \]
    and so
    \begin{align*}
    \Eop_{\cC}\left[|\{A \in \cM_{n,\tau}: A \subseteq \cC\}|\right]=|\cM_{n,\tau}|q^{-b(1-R)n}\leq q^{(H_q(\tau)-b(1-R))n} \ . & \qedhere
 \end{align*}
\end{proof}

Before concluding this section, we elucidate why a code ensemble that $L$-locally similar to random linear codes achieves the Elias bound with high probability. This is implicit in, e.g.,~\cite{MRRSW}, but we spell it out for completeness (and as prior works did not focus explicitly on achieving the Elias bound, as we do). In the remainder of the paper, we will just focus on proving our codes are $L$-locally similar to RLCs. 

Of course, one could prove a similar claim for list-recoverability; however, for this property the results for random linear codes are a bit less settled. However, it is indeed known that random linear codes achieve list-recovery capacity, and so this property will directly translate to sufficiently locally similar codes as well. Furthermore, we do emphasize that any improved result concerning the list-recoverability of random linear codes will immediately apply to our code ensembles as well. 

\begin{proposition}
    Let $q$ be a prime power. Let $L,k,n \in \N$, let $R=k/n$, and let $\eps=1/L>0$. For sufficiently large $n$ (compared to $q$ and $L$), there exists a constant $c>0$ such that the following holds. 
    
    Fix $\rho$ such that $R=1-h_q(\rho)-c\eps$, i.e., $\rho = h_q^{-1}(1-R+c\eps)$. Let $\cC\leq \F_q^n$ be a random code of rate $R$ which is $L$-locally similar to $\RLC(R)$. Then, with probability at least $1-q^{-\eps(n-o_{n\to\infty}(1))}$, $\cC$ is $(\rho,L)$-list-decodable. 
\end{proposition}

\begin{proof}
    Let $\cP$ denote the property of being $(\rho,L)$-list-decodable, which we recall from the prior discussion is $L$-local. Appealing to either~\cite[Theorem~2.4]{LiW18} (for the $q=2$ case) or~\cite[Theorem~6]{GHK11} (for the general $q$ case), we have that there exists a constant $c'>0$ such that a random linear code of rate $R'=1-h_q(\rho)-c'\eps$ is $(\rho,L)$-list-decodable with probability at least $1-\exp(-\Omega(n))$. In particular, for sufficiently large $n$, \Cref{prop:sharp-threshold} implies $R'\leq \TRLC(\cP)$.

    Let $c=c'+1$, and so $R = R'-\eps \leq \TRLC(\cP)-\eps$. \Cref{prop:prob-bound-loc-sim} then implies that $\cC$ satisfies $\cP$ with probability at least $1-q^{-\eps(n-o_{n\to\infty}(1))}$, as desired. 
\end{proof}

%% file: linearized_polynomials.tex
We now present our construction based on linearized polynomials. For a prime power $q$, block-length $n$ and target dimension $k$, fix an $\F_q$-linear subspace $V$ of $\F_{q^n}$. Fix a basis $\alpha_1,\dots,\alpha_k$ for $V$ and extend it to a basis $\alpha_1,\dots,\alpha_n$ for all of $\F_{q^n}$. We also define the $\F_q$-linear map $\varphi:\F_{q^n} \to \F_q^n$ by mapping $\alpha = x_1\alpha_1 + x_2\alpha_2 + \cdots + x_n\alpha_n \mapsto (x_1,x_2,\dots,x_n)$, where $x_1,\dots,x_n\in\F_q$.

For a target locality $\ell$, we now define a linear code by sampling $f(X) \in \F_{q^n}[X]$ uniformly at random among all linearized polynomials of $q$-degree at most $\ell-1$. That is, $f_0,f_1,\dots,f_{\ell-1}$ are sampled independently and uniformly at random from $\F_{q^n}$, and we then set $f(X)=\sum_{i=0}^{\ell-1}f_iX^{q^i}$. 

We then define the code 
\[
    \cC := \varphi(f(V)) = \{\varphi(f(\alpha))\suchthat \alpha \in V\} \ .
\]
We say that a code sampled as above is a $\PCLP(R,\ell)$ code. 

Recalling that linearized polynomials define $\F_q$-linear maps (\Cref{prop:linearized-polys-linear}), it follows readily that $\cC$ is indeed linear. Hence, to any such $\cC$ we can associate a generator matrix $G \in \F_q^{k \times n}$.

\medskip

We first observe that if $\alpha_1,\dots,\alpha_b \in V$ with $1 \leq b \leq \ell$ are linearly independent, then over the choice of $f$, $f(\alpha_1),\dots,f(\alpha_k)$ are independent and uniformly random elements of $\F_{q^n}$

\begin{lemma} \label{lem:linearized-poly-evals}
    Let $1 \leq b \leq \ell$. Suppose $\alpha_1,\dots,\alpha_b \in V$ are linearly independent and let $\beta_1,\dots,\beta_b \in \F_{q^n}$. Then 
    \[
        \Pr\left[\bigwedge_{i=1}^b f(\alpha_i)=\beta_i \right] = q^{-nb} \ .
    \]
\end{lemma}

\begin{proof}
    Writing $\bm{f} = (f_0,f_1,\dots,f_{\ell-1})^\top \in \F_{q^n}^\ell$, $\bm{\beta} = (\beta_1,\dots,\beta_b)^\top \in \F_{q^n}^b$ and defining the \emph{Moore matrix}
    \[
        M = \begin{pmatrix}
            \alpha_1 & \alpha_1^q & \cdots & \alpha_1^{q^{\ell-1}} \\
            \alpha_2 & \alpha_2^q & \cdots & \alpha_2^{q^{\ell-1}} \\
            \vdots & \vdots & \ddots & \vdots \\
            \alpha_b & \alpha_b^q & \cdots & \alpha_b^{q^{\ell-1}} \\
        \end{pmatrix} \ ,
    \]
    the event $\bigwedge_{i=1}^b f(\alpha_i)=\beta_i$ is equivalent to $M\bm{f} = \bm{\beta}$. As $b \leq \ell$ and $\alpha_1,\alpha_2,\dots,\alpha_b$ are linearly independent, the matrix $M$ has rank $b$. Hence, $M\bm{f}$ is distributed uniformly in $\F_{q^n}^b$; in particular, the probability that it takes on the value $\bm{\beta}$ is $|\F_{q^n}^b|^{-1} = q^{-nb}$.
\end{proof}

That the code $\cC$ is $\ell$-locally similar to $\RLC(R)$, where $R:=k/n$, will follows from the following lemma. 

\begin{lemma} \label{lem:prob-matrix-contained-linearized}
    Let $1 \leq b \leq \ell$ and let $A \in \F_q^{n \times b}$ be of rank $b$. Then $\Pr[A \subseteq \cC] \leq q^{-n(1-R)b}$.
\end{lemma}

\begin{proof}
    Let $G \in \F_q^{k \times n}$ be a (random) generator matrix for $\cC$. 
    From \Cref{lem:linearized-poly-evals}, we can conclude the following: if $X \in \F_q^{b \times k}$ is of rank $b$, $XG$ is uniformly random over $\F_q^{b \times n}$. Indeed, $G$ implements a linear map of the form $\varphi \circ f \circ \varphi'$ where $\varphi':\F_q^k \to V$ is an isomorphism. Fixing a basis $\bv_1,\dots,\bv_b$ for the row-space of $X$ and putting $\alpha_i=\varphi'(\bv_i)$, \Cref{lem:linearized-poly-evals} says that the tuple $(f(\alpha_1),\dots,f(\alpha_b))$ is distributed uniformly over $\F_{q^n}^b$ (as an isomorphism preserves linear independence). Hence, as $\varphi$ and $\varphi'$ are bijective, the tuple $(\varphi \circ f\circ \varphi'(\bv_1),\dots,\varphi'\circ f \circ \varphi'(\bv_b))$ is uniform over $(\F_q^n)^b$, which implies that $XG$ must be uniform over $\F_q^{b \times n}$, as desired. 
    
    We therefore have
    \begin{align}
        \Pr[A \subseteq \cC] &= \Pr[\exists X \in \F_q^{k \times b} \suchthat XG=A] = \sum_{X \in \F_q^{k \times b}}\Pr[XG=A] \nonumber \\
        &\leq \sum_{\substack{X \in \F_q^{k \times b} \\ \rank(X)=b}} \Pr[XG=A] \label{eq:rank-condn-1}\\ 
        &= \sum_{\substack{X \in \F_q^{k \times b} \\ \rank(X)=b}} q^{-nb} \leq q^{kb} \cdot q^{-nb} = q^{-(1-R)nb} \label{eq:applying-obs} \ .
    \end{align}
    In the above, \eqref{eq:rank-condn-1} follows from the fact that if $X$ has $\rank(X)<b$, then $\rank(XG)<b$, so it can't be that $XG=A$ as $\rank(A)=b$. The first equality of \eqref{eq:applying-obs} follows from the above observation. 
\end{proof}

In order to analyze the dual code $\cC^\perp$, we characterize it as a $\PCLP(1-R,1)$, making use of concepts introduced in \Cref{sec:alg-concepts}. To do this, given $f(X) = \sum_{i=0}^{\ell-1}f_i X^{q^i}$ where we assume for now that $f_0\neq 0$, let $g(X) = f_0^{-1}\cdot f(X) = X+\sum_{i=1}^{\ell-1}\tfrac{f_i}{f_0}X^{q^i}$. We then define  
\[
    W := g(V)^\perp = \{\beta \in \F_{q^n} \suchthat \forall \alpha \in g(V),~\Tr(\alpha\beta)=0\} \ ,
\]
where we recall $V$ was an $\F_q$-subspace of dimension $k$. For now, let us assume $g(V)$ has dimension $k$; below we will establish that this holds with probability $1-q^{-\Omega(n)}$. It follows then that $W$ is an $\F_q$-subspace of dimension $n-k$. 

Now, from the basis $\alpha_1,\dots,\alpha_n$ for $\F_{q^n}$, let $\beta_1,\dots,\beta_n$ be the dual basis, which we recall means 
\[
    \Tr(\alpha_i \beta_j) = \delta_{ij} \ ,
\]
where $\delta_{ij}$ is the Kronecker $\delta$-function. Finally, define the $\F_q$-linear isomorphism $\psi:\F_{q^n}\to\F_q^n$ by mapping $\beta = \sum_i y_i\beta_i \mapsto (y_1,\dots,y_n)$. We may now characterize the dual code as follows.

\begin{lemma} \label{lem:dual-characterization}
    Let $f(X)=\sum_{i=0}^{\ell-1}f_i X^{q^i} \in \F_{q^n}[X]$ be a linearized polynomial of $q$-degree at most $\ell-1$ with $f_0 \neq 0$, and let $V,W$ be as above. Then $\varphi(f(V))^\perp= \psi(f_0^{-1}W)$. 
\end{lemma}

\begin{proof}
    Note that as $\dim_{\F_q}(g(V))+\dim_{\F_q}(f_0^{-1}W)=n$, it must also hold that $\dim(\varphi(g(V))+\dim(\psi(f_0^{-1}W))=n$, as the isomorphisms $\varphi$ and $\psi$ preserve the dimension of subspaces. Hence, it suffices to show $\varphi(f(V)) \perp \psi(f_0^{-1}W)$. 

    To establish this, let $\bx = (x_1,\dots,x_n)\in \varphi(f(V))$ and $\by = (y_1,\dots,y_n) \in \psi(f_0^{-1}W)$. Then $\alpha = \sum_{i=1}^nx_i \alpha_i \in f(V)$ and $\beta = \sum_{i=1}^n y_i \beta_i \in f_0^{-1}W$. As $\beta \in f_0^{-1}W$, $f_0\beta \in W$. Similarly, $f_0^{-1}\alpha \in f_0^{-1}\cdot f(V) = g(V)$. As $W \perp g(V)$ by construction, we find 
    \[
        \Tr(\alpha \beta) = \Tr(f_0^{-1}\alpha \cdot f_0\beta) = 0 \ .
    \]
    Thus, 
    \begin{align*}
        \langle \bx,\by\rangle &= \sum_{i=1}^nx_iy_i = \sum_{i=1}^nx_iy_i\Tr(\alpha_i\beta_i)\\
        &= \sum_{i=1}^n\sum_{j=1}^nx_iy_j\Tr(\alpha_i\beta_j) = \Tr\left(\sum_{i=1}^n\sum_{j=1}^nx_iy_j\alpha_i\beta_j\right)\\
        &= \Tr\left(\left(\sum_{i=1}^nx_i\alpha_i\right)\left(\sum_{j=1}^ny_j\beta\right)\right) = \Tr(\alpha\beta)=0 \ . \qedhere
    \end{align*}
\end{proof}

Thus, assuming $f(V)$ is of $\F_q$-dimension $k$, the dual of a $\PCLP(R,\ell)$ code is essentially also a $\PCLP(1-R,1)$ code! Formally, we can derive the following. 

\begin{theorem} \label{thm:PCLP-main}
    Let $k \leq n \in \N$, $q$ a prime power. Put $R = k/n$. Let $\ell \in \N$, and let $f(X) = \sum_{i=0}^{\ell-1}f_iX^{q^i}$ where $f_0,f_1,\dots,f_{\ell-1}$ are sampled independently and uniformly at random from $\F_{q^n}$. Fix $\eps>0$. The following holds: 
    \begin{itemize}
        \item Let $\cP$ be an $\ell$-local property for which $R \leq \TRLC(\cP)-\eps$. Then, $\cP$ is also satisfied with probability at least $q^{-n(\eps - o_{n\to\infty}(1))}$ by a $\PCLP(R)$. 
        \item With probability at least $1-q^{-n(\eps - o_{n\to\infty}(1))}$, the dual distance of a $\PCLP(R)$ is at least $h_q^{-1}(R-\eps)$. 
        \item With probability at least $1-q^{-(1-R)n}$, the code $\cC$ has rate $R$.
    \end{itemize}
    In particular, all items hold with probability at least $1-q^{-n(\eps - o_{n\to\infty}(1))}$. 
\end{theorem}

\begin{proof}
    For the first bullet-point, combining \Cref{lem:prob-matrix-contained-linearized} and \Cref{prop:prob-bound-loc-sim}, we find that a $\PCLP(R,\ell)$ code is $\ell$-locally similar to $\RLC(R)$. The desired result is then a direct consequence of the ``In particular'' part of \Cref{prop:prob-bound-loc-sim}.

    For the second bullet-point, consider the following sampling procedure:
    \begin{itemize}
        \item First, sample a $g(X) \in \F_{q^n}[X]$ to be a linearized polynomial of $q$-degree at most $\ell-1$ with the additional constraint that $g_0=1$, and subject to $g(V)$ having dimension $k$.
        \item Define then $W := g(V)^\perp$, which is of dimension $n-k = (1-R)n$ over $\F_q$.
        \item Sample $a \in \F_{q^n}$ uniformly at random.
        \item Output the code $\psi(aW)$. 
    \end{itemize}
    We remark that such a code is distributed identically to a $\PCLP(1-R,1)$ code. By the first bullet-point, it thus satisfies any $1$-local property that is satisfied with high probability by a $\RLC(1-R)$. 
    
    In particular, we can easily prove it has minimum distance at least $h_q^{-1}(R-\eps)$ with probability at least $1-q^{-\eps n}$. Formally, if $\cB$ is a code sampled as above and we set $\delta = h_q^{-1}(R-\eps)$ and $S = \{\bx \in \F_q^n:0<\wt(x) \leq \delta\}$, we have 
    \begin{align}
        \Pr[\exists \bc \in \cB \text{ such that } \wt(\bc) \leq \delta] &\leq \sum_{\bx \in S}\Pr[\bx \in \cB] \label{eq:union-bound}\\ 
        &\leq \sum_{\bx \in \cB} q^{-Rn} \label{eq:applying-lemma} \\
        &\leq q^{h_q(\delta)n}q^{-Rn} = q^{n(R-\eps) - Rn} = q^{-\eps n} \ .\label{eq:bound-on-ball}
    \end{align}
    Above, \eqref{eq:union-bound} applies the union bound; \eqref{eq:applying-lemma} applies \Cref{lem:prob-matrix-contained-linearized} (with $\ell=1$ and rate $1-R$) and the first inequality of \eqref{eq:bound-on-ball} is a standard estimate on the number of vectors of Hamming weight at most $\delta$. 

    Now, in the actual sampling procedure for a $\PCLP(R,\ell)$ code $\cC$, conditioned on the event $E_1$ that $f(V)$ is of $\F_q$-dimension $k$, and the event $E_2$ that $f_0\neq0$, we find that $\cC^\perp$ is distributed identically to $\cB$. Clearly, $\Pr[E_2]=1-q^{-n}$. 

    For $\Pr[E_1]$, note that $f(V)$ is of $\F_q$-dimension $k$ if and only if $\ker(f)\cap V=\{0\}$. Fixing a basis $\alpha_1,\dots,\alpha_k$ for $V$ of $\F_q$, note that $\ker(f)\cap V \neq \{0\}$ if and only if there exists $x_1,\dots,x_k \in \F_q$, not all $0$, for which $\alpha=\sum_{i=1}x_i\alpha_i \in \ker(f)$. That is, $f(\alpha)=0$. Denoting by $\bx = (x_1,\dots,x_k)$, $\bm{f}=(f_0,f_1,\dots,f_{\ell-1})^\top \in \F_{q^n}^\ell$ and $M \in \F_{q^n}^{n \times \ell}$ the Moore matrix defined via $M_{ir} = \alpha_i^{q^{r-1}}$, we have that $f(\alpha)=0$ if and only if 
    \[
        \bx M\bm{f} = \bm{0} \ .
    \]
    If $\bm{f}$ is uniformly random, then $M\bm{f}$ is a uniformly random vector in $\F_{q^n}^k$, and hence it is orthogonal to $\bx$ (which we here view as an element of $\F_{q^n}^k$) with probability $q^{-n}$. By taking a union bound over $q^k-1\leq q^k$ choices for $\bx$, we conclude that $\Pr[\neg E_1]\leq q^{k-n} = q^{-(1-R)n}$
    
    Hence, if $A$ denotes the event that $\cC^\perp$ has distance at least $h_q^{-1}(R-\eps)$ and $B$ denotes the event that a $\B$ code has distance at least $h_q^{-1}(R-\eps)$, we find
    \begin{align*}
        \Pr[B] &\geq \Pr[E_1\land E_2] \Pr[B|E_1\land E_2] \geq (1-\Pr[\neg E_1] - \Pr[\neg E_2])\Pr[A] \\
        &\geq 1-q^{-(1-R)n} - q^{-n} - q^{-\eps(n-o_{n \to \infty}(1))} \geq 1-q^{-\eps(n-o_{n \to \infty}(1))}\ ,
    \end{align*}
    as claimed. 

    The final bullet-point also follows from the above argument: if the event $E_1$ occurs, so $f(V)$ is of $\F_q$-dimension $k$, then certainly the code $\cC$ will have size $q^k = q^{Rn}$, and hence rate $R$. As we just argued $\Pr[E_1]\geq 1-q^{-(1-R)n}$, the last bullet-point follows.
\end{proof}
Before concluding this section, we state a couple more pleasant properties of $\PCLP(R,\ell)$ codes. Of particular note, we indicate that these codes can be encoded in quasi-linear time.

\begin{proposition} \label{prop:PCLP-sampling-and-encoding}
    Let $k \leq n \in \N$, $q$ a prime power. Put $R = k/n$ and let $\ell \in \N$. A $\PCLP(R,\ell)$ code can be sampled with 
    \[
        \ell \lceil n \log_2q\rceil
    \]
    uniformly random bits. 
    
    Furthermore, given $f(X) = \sum_{i=1}^nf_iX^{q^i}$ we can choose $V$ with a specific basis $\alpha_1,\dots,\alpha_k$ such that the encoding algorithm of a $\PCLP(R,\ell)$ code can be implemented in $O(n\log n)$ time.
\end{proposition}

\begin{proof}
    The randomness bound is immediate. For the quasi-linear encoding, let us fix a degree $n$ irreducible polynomial $m(X)$ and work with the representation $\F_{q^n} \cong \F_q[X]/(m(X))$. Let $\lambda$ be a root of $m(X)$; it is well-known that then $1,\lambda,\dots,\lambda^{n-1}$ is a basis for $\F_{q^n}$ over $\F_q$. Put $\alpha_i = \lambda^{i-1}$ for $i\in[n]$, and define the isomorphism $\iota:\F_q^k \to V = \mathrm{span}\{1,\lambda,\dots,\lambda^{k-1}\}$ via $\iota(x_1,\dots,x_k) = \sum_{i=1}^k x_i\lambda^{i-1}$. We wish to show that we can implement the linear map $\varphi\circ f\circ \iota:\F_q^k \to \F_q^n$ in $O(n\log n)$ time. 

    For $i=0,\dots,\ell-1$ let $g_i(X) \in \F_q[X]$ be polynomials of degree at most $n-1$ for which $f_i = g_i(\lambda)$. (Such polynomials come naturally from our representation $\F_{q^n} \cong \F_q[X]/(m(X))$.) Next, given a message $\bx = (x_1,\dots,x_k) \in \F_q^k$, let $h_0(X) = \sum_{i=0}^{k-1}x_iX^i$, and note then that $h_0(\lambda) = \iota(\bx)$. Similarly, we define polynomials $h_i(X) \in \F_q[X]$ of degree at most $n-1$ s.t. $h_i(X) = h_0(X)^{q^i} \mod m(X)$. Note that given $h_i(X)$, one can compute $h_{i+1}(X)$ in time $O(n\log n)$: indeed, if 
    \[
        h_i(X) = \sum_{j=0}^{n-1}h_{ij}X^i, ~ h_{ij} \in \F_q \ ,
    \]
    then $h_{i+1}(X) = \sum_{j=0}^{n-1}h_{ij}X^{iq}$, and then we need to reduce it modulo $m(X)$ which can be done in $O(n\log n)$ time via the Euclidean algorithm. Observe that this implies that $h_i(\lambda) = h_0(\lambda)^{q^i} = (\iota(\bx))^{q^i}$.

    Now, from the polynomials $g_0(X),g_1(X),\dots,g_{\ell-1}(X)$ and $h_0(X),h_1(X),\dots,h_{\ell-1}(X)$, we note that the encoding of $\bx$ -- which is $\varphi\circ f\left(\sum_{i=1}^k x_i \lambda^i\right)$ -- is defined in terms of all of the pairwise products $g_j(\lambda)h_j(\lambda)$ for all $j \in \{0,1,\dots,\ell-1\}$. Indeed, note that
    \[
        f\left(\iota(\bx)\right) = \sum_{j=0}^{\ell-1}f_j \cdot \iota(\bx)^{q^j} = \sum_{j=0}^{\ell-1}g_j(\lambda) \cdot h_j(\lambda) \ .
    \]
    To compute these products, we can perform $\ell=O(1)$ multiplications of the polynomials $g_j(X)h_j(X)$, each of which can be performed in $O(n\log n)$ applying a FFT, followed by a reduction modulo $m(X)$ which can be implemented in $O(n\log n)$ time as well. Upon observing that the coefficients of the resulting degree at most $n-1$ polynomial in $\F_q[X]$ yield the desired encoding of $\bx$ in $\F_q^n$, the proof is complete. 
\end{proof}

\subsection{Generalized Wozencraft Ensemble} \label{subsec:generalized-wozencraft-ensemble}

Before continuing, we remark that linearized polynomials also yield a pleasant generalization of the well-known \emph{Wozencraft ensemble}~\cite{M63}. We recall that such codes are obtained by randomly choosing $\beta \in \F_{q^k}$ and then defining
\[
    \cC := \{(\varphi(\alpha),\varphi(\beta \alpha)):\alpha \in \F_{q^k} \} \ ,
\]
where $\varphi:\F_{q^k} \to \F_q^k$ is a full-rank $\F_q$-linear map. Such codes are well-known to achieve the GV-bound with high probability; in fact, the argument shows that they are $1$-locally-similar to random linear codes of rate $1/2$. 

Defining $f(X) = \beta X$, which is a $q$-degree $0$ linearized polynomial, one can view the codewords of $\cC$ as $(\varphi(\alpha),\varphi(f(\alpha))$ for $\alpha \in \F_{q^k}$. It is natural to generalize this construction to allow $f$ to have $q$-degree at most $\ell-1$, and if $f$ is chosen uniformly at random among all such linearized polynomials an argument analogous to the one given above would demonstrate that this generalization of Wozencraft's ensemble is $\ell$-locally similar to $\RLC(1/2)$. 

We prefer our construction as it naturally allows for all choices of rate. However, we point out this construction now, as it technically only requires $\ell \lceil k\log_2q\rceil = \ell \lceil (n-k)\log_2q\rceil$ random bits to sample this code, where $k=n/2$ is the dimension. In light of \Cref{prop:randomness-lb}, this is actually \emph{optimal} if one hopes for a code distribution which is $\ell$-locally similar to $\RLC(1/2)$.\footnote{Up to the ceilings, which can be removed if $q$ is a power of $2$.} Hence, the lower bound is actually achievable. 

We can also generalize this construction slightly by defining
\[
    \cC_r := \left\{\big(\varphi(\alpha),\varphi(f_1(\alpha),\varphi(f_2(\alpha),\ldots,\varphi(f_r(\alpha)\big):\alpha \in \F_{q^k} \right\} \ ,
\]
where $f_1(x),\ldots,f_r(x)$ are random linearized polynomials of $q$-degree at most $\ell-1$ which are distributed mutually independent and uniformly at random. It is clear that the rate of $\cC_r$ is $1/r$. A similar argument to that given above proves that $\cC_r$ is $\ell$-locally similar to $\RLC(1/r)$. Since we only need $\ell \lceil(n-k)\log_2q\rceil$ random bits to sample this code, the lower bound $(1-R)n\log_2 q$ is also achievable for rates $R=\frac{1}{r}$ where $r\geq 2$ is an integer.

%% file: new_constructions_nic.tex

We now present our construction of pseudorandom codes from row and column polynomials, which not only satisfy $\ell$-local similarity to RLCs, but their \emph{duals} are also $\ell$-locally similar to RLCs.

Fix a locality parameter $\ell \in \N$, a block-length $n$ and a rate $0 < R < 1$ such that $k := Rn$ is an integer, and a prime power $q$. For a basis $\omega_1,\dots,\omega_n$ for $\F_{q^n}$ over $\F_q$, define the full-rank $\F_q$-linear map $\varphi:\F_{q^n} \to \F_q^n$ as $\varphi:\sum_{i=1}^nx_i \omega_i = (x_1,\dots,x_n)$. Similarly, define the surjective $\F_q$-linear map $\psi: \F_{q^n} \to \F_q^k$ as $\psi:\sum_{i=1}^nx_i \omega_i = (x_1,\dots,x_k)$. Finally, let $\alpha_1,\dots,\alpha_n \in \F_{q^n}$ be $n$ distinct elements.

Let $f(X),g(X) \in \F_{q^n}[X]$ be two random polynomials, sampled uniformly and independently amongst all polynomials of degree at most $\ell-1$. That is, $f(X) = \sum_{i=0}^{\ell-1}f_iX^i$ and $g(X) = \sum_{i=0}^{\ell-1}g_iX^i$, where the coefficients $\{f_i\}_{i=0}^{\ell-1}$ and $\{g_i\}_{i=0}^{\ell-1}$ are sampled independently and uniformly from $\F_{q^n}$. 

We now define the matrices $G',G'' \in \F_q^{k \times n}$ as follows: 
\begin{itemize}
    \item For $i \in [k]$, the $i$-th row of $G'$ is $\varphi(f(\alpha_i))$.
    \item For $j \in [n]$, the $j$-th column of $G''$ is $\psi(g(\alpha_j))$. 
\end{itemize}
We then define $G := G'+G''$, and take $\cC := \{\bx G\suchthat \bx \in \F_q^k\} \leq \F_q^n$. We say that a code sampled as above is a $\PCRCP(R,\ell)$ code. 

\begin{lemma} \label{lem:matrix-dist-6}
    Let $1 \leq b \leq \ell$ be an integer. The following hold.
    \begin{enumerate}
        \item Fix a matrix $X \in \F_q^{b \times k}$ of rank $b$. Then $XG$ is uniformly distributed over $\F_q^{b \times n}$. 
        \item Fix a matrix $X \in \F_q^{n \times b}$ of rank $b$. Then $GX$ is uniformly distributed over $\F_q^{k \times b}$. 
    \end{enumerate}
\end{lemma}

\begin{proof}
    We prove each item in turn, beginning with the first. Fix a matrix $Y \in \F_q^{n \times b}$, and note that it suffices to argue that $\Pr_{G'}[XG'=Y]=q^{-nb}$. Indeed, having established this one can condition on the choice of $G''$, which is probabilistically independent of $G'$, to conclude the desired result as $XG = XG' + XG''$, and so $XG' = XG-XG''$.

    Let $\by_1,\dots,\by_b \in \F_q^n$ denote the rows of $Y$. Note that the event $XG'=Y$ is equivalent to 
    \[
        \forall 1\leq j \leq b, ~~ \sum_{i=1}^k X_{ji}\varphi(f(\alpha_i)) = \by_j \ .
    \]
    By $\F_q$-linearity of $\varphi$, this is equivalent to 
    \[
        \forall 1\leq j \leq b, ~~ \varphi\left(\sum_{i=1}^k X_{ji}f(\alpha_i)\right) = \by_j \ .
    \]
    Defining $z_j := \varphi^{-1}(\by_j) \in \F_{q^n}$ for $1 \leq j \leq b$ (recall that $\varphi:\F_{q^n} \to \F_q^n$ is chosen to be full-rank, hence bijective), the above holds if and only if 
    \begin{align}\label{eq:after-varphi-inv}
        \forall 1\leq j \leq b, ~~ \sum_{i=1}^k X_{ji}f(\alpha_i) = z_j \ .
    \end{align}
    Note now that, for each $j \in [b]$,
    \begin{align*}
        \sum_{i=1}^k X_{ji}f(\alpha_i) &= \sum_{i=1}^kX_{ji} \sum_{r=0}^{\ell-1} f_r\alpha_i^r = \sum_{r=0}^{\ell-1}f_r \sum_{i=1}^k X_{ji}\alpha_i^r \ .
    \end{align*}
    Defining now the Vandermonde matrix $V \in \F_{q^n}^{k \times \ell}$ as $V_{ir} := \alpha_i^r$, and letting $\bm{f}:=(f_0,\dots,f_{\ell-1})^\top \in \F_{q^n}^\ell$ and $\bz := (z_1,\dots,z_b)^\top \in \F_{q^n}^b$, we observe that \Cref{eq:after-varphi-inv} holds if and only if 
    \begin{align*}
        XV\bm{f} = \bz \ .
    \end{align*}
    As $V$ is a Vandermonde matrix and $\alpha_1,\dots,\alpha_k$ are distinct, it is of rank $b$ over $\F_{q^n}$. Moreover, as $X$ is of rank $b$ over $\F_q$, it is also of rank $b$ over $\F_{q^n}$, for if it were of lower rank then some $b \times b$ minor would be $0$, implying that $X$ was already of rank less than $b$ over $\F_q$. Hence, the linear map defined by $XV$ from $\F_{q^n}^\ell \to \F_{q^n}^b$ is $q^{n(\ell-b)}$-to-1: in particular, it transports the uniform distribution from $\F_{q^n}^\ell$ to $\F_{q^n}^b$. As $\bm{f}$ is uniformly random over $\F_{q^n}^\ell$, it follows $XV\bm f$ is uniformly random over $\F_{q^n}^b$; in particular, the probability that $XV\bm f = \bz$ is $(q^n)^{-b} = q^{-nb}$. Hence, $\Pr[XG'=Y]=q^{-nb}$, as desired. 

    We now turn to the second item, whose proof is completely analogous. In this case, we show that for any $Y \in \F_q^{k \times b}$, $\Pr_{G''}[G''X=Y] = q^{-kb}$. For reasons analogous to before, this suffices to derive the claim. 

    Again, let $\by_1,\dots,\by_b \in \F_q^n$ denote the columns of $Y$. This time, we note that $G''X=Y$ is equivalent to 
    \[
        \forall 1 \leq j \leq b, ~~\sum_{i=1}^nX_{ij}\psi(g(\alpha_i)) = \by_j \ , 
    \]
    which by $\F_q$-linearity of $\psi$ is equivalent to 
    \[
        \forall 1 \leq j \leq b, ~~\psi\left(\sum_{i=1}^n X_{ij} g(\alpha_i)\right) = \by_j \ .
    \]
    In this case, $\psi$ is not injective. However, if $K = \ker(\psi)$ (which is the $\F_q$-span of $\{\omega_{k+1},\dots,\omega_n\}$) and $z_j \in \F_{q^n}$ are any elements for which $\psi(z_j) = \by_j$ (such elements exist as $\psi$ is surjective), the above is equivalent to the following condition:
    \begin{align}\label{eq:after-psi-inv}
        \forall 1 \leq j \leq b, ~~\sum_{i=1}^n X_{ij}g(\alpha_i) \in z_j + K \ .
    \end{align}
    Now, note that for all $j \in [b]$,
    \[
        \sum_{i=1}^n X_{ij}g(\alpha_i) = \sum_{i=1}^n \inparen{X_{ij}\sum_{r=0}^{\ell-1}g_r\alpha_i^r} = \sum_{r=0}^{\ell-1}\inparen{g_r\sum_{i=1}^n X_{ij}\alpha_i^r} \eperiod
    \]
    Hence, defining now the Vandermonde matrix $W \in \F_{q^n}^{\ell \times n}$ as $W_{ri} :=\alpha_i^r$, and letting $\bm{g}:=(g_0,\dots,g_{\ell-1}) \in \F_{q^n}^\ell$ and $\bm{K}(\bm{z}):=(K+z_1) \times(K+z_z) \times \cdots \times (K+z_n)$, we find that \Cref{eq:after-psi-inv} holds if and only if 
    \[
        \bm{g}WX \in \bm{K}(\bm{z}) \ .
    \]
    As before, we have that $WX$ is of rank $b$ (even over $\F_{q^n}$), and thus $\bm{g}WX$ is uniformly random over $\F_{q^n}^b$. As $|\bm{K}(\bm{z})| = q^{b(n-k)}$, it follows that the probability that $\bm{g}WX \in \bm{K}(\bm{z})$ is $\frac{q^{b(n-k)}}{q^{bn}} = q^{-bk}$. Hence, $\Pr_{G''}[G''X = Y] = q^{-kb}$, as claimed. 
\end{proof}

From here, the $\ell$-local similarity to RLCs for both the primal and dual code follows readily. This yields our main result for this section.

\begin{theorem} \label{thm:row-column-polynomial-codes}
    Let $\ell, k, n \in \N$ with $k \leq n$, put $R = k/n$ and let $q$ be a prime power. Let $\cC$ be a $\PCRCP(R,\ell)$ code. Then it is $\ell$-locally similar to $\RLC(R)$, and also $\cC^\perp$ is $\ell$-locally similar to $\RLC(1-R)$. 

    Furthermore, the code can be sampled with $2\ell \lceil n\log_2q\rceil$ uniformly random bits. Lastly, with probability at least $1-q^{-(1-R)n}$, the code has rate $R$. 
\end{theorem}

\begin{proof}
    The claim regarding the randomness complexity is immediate. As for the rate, we note that the rate of $\cC$ is less than $R$ if and only if there exists some $\bm{m} \in \F_q^k\setminus\{\bm{0}\}$ for which $\bm{m}G=\bm{0}$. Applying \Cref{lem:matrix-dist-6} with $b=1$, it follows that $\bm{m}G$ is distributed uniformly over $\F_q^n$, and hence it takes on value $\bm{0}$ with probability $q^{-n}$. Upon taking a union bound over the $q^k-1\leq q^k$ choices for $\bm{m}$, the claim follows.

    We now turn to the $\ell$-local similarity claims. Let $A \in \F_q^{n \times b}$ of rank $b$. Appealing to \Cref{prop:prob-bound-loc-sim}, it suffices to show the following two bounds:
    \begin{itemize}
        \item $\Pr[A \subseteq \cC] \leq q^{-n(1-R)b}$. 
        \item $\Pr[A \subseteq \cC^\perp] \leq q^{-nRb}$. 
    \end{itemize}
    We start with the first item. We have 
    \begin{align}
        \Pr[A \subseteq \cC] &= \Pr[\exists X \in \F_q^{k \times b} \text{ such that } XG=A] = \sum_{X \in \F_q^{k \times b}}\Pr[XG=A] \nonumber \\
        &\leq \sum_{\substack{X \in \F_q^{k \times b} \\ \rank(X)=b}} \Pr[XG=A] \label{eq:rank-condn}\\ 
        &= \sum_{\substack{X \in \F_q^{k \times b} \\ \rank(X)=b}} q^{-nb} \leq q^{kb} \cdot q^{-nb} = q^{-(1-R)nb} \label{eq:applying-lem-6.1a} \ .
    \end{align}
    In the above, as in the proof of \Cref{thm:PCLP-main}, the inequality~\eqref{eq:rank-condn} follows from the fact that if $X$ has $\rank(X)<b$, then $\rank(XG)<b$, so it can't be that $XG=A$ as $\rank(A)=b$. The first equality of \Cref{eq:applying-lem-6.1a} follows the first item of \Cref{lem:matrix-dist-6}. 

    We now address the second item. First, note that $\bx \in \cC^\perp \iff G\bx^\top = \bm{0}$, and hence $A \subseteq \cC^\perp \iff GA=0$. We have  
    \begin{align*}
        \Pr[A \subseteq \cC^\perp] = \Pr[GA=0] = q^{-kb} = q^{-Rnb} \ ,
    \end{align*}
    where we applied the second item of \Cref{lem:matrix-dist-6} in the second equality. This establishes the second claim, and hence completes the proof of the theorem. 
\end{proof}

%% file: sublinear_randomness.tex
In this section, we elucidate why $\Omega(n\ell\log_2(q))$ randomness, as our constructions achieve, is currently a challenging barrier to overcome. In particular, we show that this is necessary for an ensemble to be $\ell$-locally similar to random linear codes. 

\begin{proposition} \label{prop:randomness-lb}
    Let $q$ be a prime power, and let $n,\ell \in \N$ such that $\frac{n}{\log_qn} \geq \omega_{n \to \infty}(q^{2\ell})$. Let $\cC$ be a code ensemble outputting block-length $n$ codes of rate $R$ that are $\ell$-locally similar to RLCs. Then, any sampling procedure for $\cC$ requires at least 
    \[
        \lceil\ell(1-R-o_{n\to\infty}(1))n\log_2q\rceil = \Omega(n\ell\log_2q)
    \]
    bits of randomness. 
\end{proposition}

\begin{proof}
    By assumption, it holds that for any $\tau \sim \F_q^\ell$ with $\dim(\tau)=\ell$, we must have 
    \[
        \Eop_{\cC}[|\{A\in \cM_{n,\tau} \suchthat A \subseteq \cC\}|] \leq q^{n(H_q(\tau)-\ell(1-R))} \ .
    \]
    We may choose $\tau$ so that $\tau(\bv) \cdot n \in \{0,1,\dots,n\}$ for all $\bv \in \F_q^\ell$, which implies $\cM_{n,\tau} \neq \emptyset$. Hence, by averaging there exists a matrix $A \in \cM_{n,\tau}$ for which 
    \[
        \Pr_{\cC}[A \subseteq \cC] \leq \frac{q^{n(H_q(\tau)-\ell(1-R))}}{|\cM_{n,\tau}|} \leq n^{O(q^\ell)} \cdot q^{-\ell(1-R)n} \ ,
    \]
    where the last inequality applies \Cref{eq:estimate-on-type-class}. By assumption on $n, q,\ell$, this latter term may be upper bounded by 
    \[
        q^{-\ell(1-R-o_{n\to\infty}(1))n} \ .
    \]
    As the sampling procedure for $\cC$ is such that some event occurs with probability at most $q^{-\ell(1-R-o_{n\to\infty}(1))n}$, it must be that the procedure consumes at least 
    \[
        \lceil\ell(1-R-o_{n\to\infty}(1))n\log_2q\rceil = \Omega(n\ell\log_2q)
    \]
    random bits, as claimed. 
\end{proof}

Recalling \Cref{subsec:generalized-wozencraft-ensemble}, this lower bound is in fact achievable (at least, for certain values of $R$ and $q$). 

This motivates the tantalizing open problem that we leave for future work: provide a randomized construction of a code achieving, say, the GV bound, using $o(n)$ bits of randomness. Such a code would be necessarily locally-\emph{dissimilar} from random linear codes. As a formally easier task, we view a construction with sublinear randomness lying on the GV bound as an interesting and useful stepping stone towards the task of providing an explicit construction.